\documentclass{imsart}

\RequirePackage[OT1]{fontenc}
\RequirePackage{amsthm,amsmath,amssymb,mathtools,bm}
\RequirePackage{graphicx,comment}
\RequirePackage[numbers]{natbib}
\RequirePackage[colorlinks,citecolor=blue,urlcolor=blue]{hyperref}
\usepackage[hmargin=0.8in,margin=1in]{geometry}
\usepackage{subfig}
\usepackage{epstopdf}
\usepackage{tabularx}

\arxiv{math.PR/00000000}
\startlocaldefs
\numberwithin{equation}{section}
\theoremstyle{plain}
\newtheorem{thm}{Theorem}[section]

\newtheorem{lem}[thm]{Lemma}
\newtheorem{prop}[thm]{Proposition}

\newtheorem{rmk}[thm]{Remark}

\endlocaldefs

\newcommand{\be}{\begin{equation}}
\newcommand{\ee}{\end{equation}}
\newcommand{\ben}{\begin{equation*}}
\newcommand{\een}{\end{equation*}}

\newcommand{\bq}{\begin{eqnarray}}
\newcommand{\eq}{\end{eqnarray}}
\newcommand{\ex}{\mathbb{E}}
\newcommand{\half}{\frac{1}{2}}

\newcommand{\ind}{\mathbf{1}}%}{1\hspace{-2.1mm}{1}} %Indicator Function
\newcommand{\nn}{\nonumber}

\newcommand{\ul}{\underline}
\newcommand{\ol}{\overline}
\newcommand{\mc}{\mathcal}
\newcommand{\p}{\partial}

\newcommand{\et}{\mathrm{e}}

\newcommand{\pr}{\mathbb{P}}

\newcommand{\diff}{\textup{d}}
\newcommand{\prq}{{\Phi(r+q)}}

\newcommand{\R}{\mathbb{R}}
\newcommand{\timset}{\mathcal{T}}
\newcommand{\U}{~\cup~}

\usepackage[normalem]{ulem}
\usepackage{color}

\hypersetup{pdfpagemode=UseNone}
\hypersetup{pdfstartview=FitH}

\renewcommand{\bar}{\overline}
\renewcommand{\tilde}{\widetilde}
\renewcommand{\hat}{\widehat}

\begin{document}

\begin{frontmatter}
\title{When to sell an asset amid anxiety about drawdowns}
\runtitle{Selling amid anxiety about drawdowns}
%\thankstext{T1}{Footnote to the title with the `thankstext' command.}

\begin{aug}
\author{\fnms{Neofytos} \snm{Rodosthenous}\thanksref{t3}%\thanksref{t1,t2}
\ead[label=e1]{n.rodosthenous@qmul.ac.uk}}
%\and
%\author{\fnms{Second} \snm{Author}\thanksref{t3}\ead[label=e2]{second@somewhere.com}}

\address{School of Mathematical Sciences, Queen Mary University of London\\ Mile End Road, London E1 4NS, UK\\
\printead{e1}}

\author{\fnms{Hongzhong} \snm{Zhang}
%\ead[label=e3]{hz2244@columbia.edu}
%\ead[label=u1,url]{www.foo.com}
}

\address{Department of IEOR, Columbia University\\
500W 120th Street, New York, NY 10027, USA\\
%\printead{e3}
%\printead{u1}
}

%\thankstext{t1}{Some comment}
%\thankstext{t2}{First supporter of the project}
\thankstext{t3}{Corresponding author.}
\runauthor{Rodosthenous and Zhang}

\affiliation{Queen Mary University of London and Columbia University}
\end{aug}

\begin{abstract}
We consider risk averse investors with different levels of anxiety about asset price drawdowns. The latter is defined as the distance of the current price away from its best  performance since inception. These drawdowns can increase either continuously or by jumps, and will contribute towards the investor's overall impatience %anxiety 
when breaching the investor's private tolerance level. We investigate the unusual reactions of investors when aiming to sell an asset under such adverse market conditions. 
Mathematically, we study the optimal stopping of the utility of an asset sale with a random discounting that captures the investor's overall impatience. %anxiety. 
The random discounting is given by the cumulative amount of time spent by the drawdowns in an undesirable high region, fine tuned by the investor's personal tolerance and anxiety about drawdowns. 
We prove that in addition to the traditional take-profit sales, the real-life employed stop-loss orders and trailing stops may become part of the optimal selling strategy, depending on different personal characteristics.
%levels of tolerance and anxiety about drawdowns. 
This paper thus provides insights on the effect of anxiety and its distinction with traditional risk aversion on decision making.
\end{abstract}

\begin{keyword}[class=MSC]
\kwd{60G40}
\kwd{60G51}
\kwd{91G80}
\end{keyword}

\begin{keyword}
\kwd{Drawdown} \kwd{L\'evy process} \kwd{optimal stopping} \kwd{Omega clock} \kwd{random discount rate} \kwd{stop-loss orders} \kwd{trailing stops}
\end{keyword}
\end{frontmatter}

\section{Introduction}

There are many economic, financial and psychological (behavioural) reasons that drive investors to panic when their asset prices experience a relatively large fall or have a relatively low value for a significant amount of time. This often results in the assets being sold at a lower than anticipated price. 
A common scenario, that leads investors to such decisions, is financial markets acting contrary to an investor's expectation for a lot longer than their investment capital can hold out. 
Another scenario is when asset prices remain at low levels, supervisors of traders grow impatient and ask ``How much longer do we have to carry this trade before it profits?''.
One can also consider the case of Commodity Trading Advisors who determine various rules for the magnitude and duration of their client accounts' drawdowns. These can be accounts that are shut down either when certain drawdowns are breached or after small long-lasting drawdowns (see, e.g. \cite{chekhlovDD} for more details). 
In all these examples, investors are driven to liquidate assets at undesirable low prices.

This phenomenon of ``selling low'' is also evident in financial markets through the extensive use of traditional {\it stop-loss} and {\it trailing stop} orders placed by investors.  
Stop-loss orders are usually placed to minimise risks associated with trading accounts, to bound losses or to protect profits, and they have a fixed value.\footnote{See \cite{Fischbacher2017} for a recent empirical study stressing the importance of stop-loss orders and their usefulness for reducing investors' disposition effect.} 
Trailing stops serve a similar purpose, but contrary to stop-loss orders, they automatically follow price movements, e.g. the asset's best historical performance.  
Even though such strategies sell a ``losing" investment without guaranteeing that  this is better than holding onto the assets, they are frequently used in practice mainly due to investors' anxiety of incurring further losses. 
Besides their aforementioned purpose, placing trailing stop orders has the additional benefit of allowing investors to focus on multiple open positions at the same time, thanks to their special self-adjustment feature (see \cite{Glynn1995} for a study of trailing stop strategies, their optimal value, duration and distribution of gains). 
Recently, the use of these strategies was also studied in various mathematical frameworks. 
For instance, the optimal combination of an up-crossing target price and a stop-loss, chosen from a specific given set, was studied in \cite{Zhang2001a} under a switching geometric Brownian motion model. 
An investigation in \cite{imkeller2014trading} under a drifted Brownian motion showed that the possibility of a negative drift also suggests the use of combinations of such strategies. 
Moreover, buying and selling strategies with exogenous trailing stops were considered in conjunction with take-profit orders in \cite{Tim_HZ18} under diffusion models.   
In practice, such stop-loss and trailing stop orders are available for use on stock, option and futures exchanges.

Although trailing stops and stop-loss orders have been widely used, there is no quantitative model that can explain the rationale behind such practices. 
In all aforementioned studies \cite{Zhang2001a,imkeller2014trading,Tim_HZ18,Zhang2018}, the use of trailing stops is exogenously imposed. 
On the other hand, Russian options and their extensions (see e.g. \cite{RZ17watermark}) involve the use of trailing stops, but as in regret theory, their objective is to protect against drawdowns, and is not concerned about the utility realised from an asset sale. The purpose of this paper is thus to provide a framework that rationalises the use of these types of orders from the perspective of selling an asset.
In particular, we neither impose any exogenous hard constraints on the set of selling strategies, nor do we have a reward that involves the running maximum of the asset's price. 
In such case, a traditional take-profit sale is usually optimal. 
However, we demonstrate in the present paper that, the use of stop-loss and trailing stop strategies naturally arises from the growing ``anxiety'' of the investor, as the asset price's performance remains at undesirable levels. 
A decision making process affected by (asset price) path-dependence can also be found, via the so-called history-dependent risk aversion models, in \cite{Dillenberger15} and references therein. 
Contrary to our work, the path-dependence affects directly the risk aversion in these models, e.g. it is incorporated in a state-dependent risk aversion coefficient $\rho$ of utility functions as in \eqref{Uu}. 
Our model can be seen as an expansion of this class of path-dependent risk aversion studies in the novel direction that we rigorously present in the sequel.

To fix ideas, consider a financial market with a risky asset whose price process $\et^X$ is modelled by a spectrally negative L\'evy process  $X=(X_t)_{t\ge0}$ on a filtered probability space $(\Omega, \mathcal{F}, \mathbb{F}=(\mathcal{F}_t)_{t\ge0},\pr)$. 
Empirical evidence suggests that such a financial market model, allowing for negative asset price jumps, is appropriate in various settings, such as equity, fixed income and credit risk (see \cite{CarrWu03} and \cite{MadanSchoutens08} among others).
In order to model the aforementioned anxiety of investors, consider also the best performance of the asset $\et^{\overline{X}}$, where 
$\overline{X}=(\overline{X}_t)_{t\geq 0}$ is the running maximum process 
associated with $X$, given by 
$\overline{X}_t = s \vee \sup_{u\in[0,t]} X_u$. 
Namely, the best performance until time $t$ is the maximum between the highest price of the asset during the time interval $[0,t]$ and the constant $s\in\R$. 
The latter represents the ``starting maximum" of the asset price and can be interpreted as the highest asset price over some previous time period $(-t_0, 0]$, for some $t_0>0$.

We model the driving concern of investors with an ``impatience" clock, when the asset price remains at relatively low levels for a significant amount of time. 
Specifically, for any fixed constant $q>0$, we construct an (impatience) Omega clock, which measures the amount of time that $X$ is below its running maximum $\overline{X}$ 
by a pre-specified level $c>0$:
\be
\varpi_t^c:=q\int_{0}^{t}\ind_{\{X_u < \overline{X}_u - c\}}\diff u.
\label{eq:omega}
\ee
In other words, it measures the amount of time that the drawdowns of the logarithm 
of the asset price $X$, away from its best performance $\overline{X}$, exceed a 
{\it tolerance level} $c$. The extent of the investors' impatience is tuned by the value of the {\it anxiety rate} $q$. 
Then, the investor is faced with a random penalisation in the form of discounting driven by \eqref{eq:omega}. 
Under the aforementioned impatience mechanism, the investor is looking for the optimal selling strategy in order to maximise the utility $U(\et^{X_\tau})$ from selling the asset  at time $\tau$.  
Mathematically, the investor aims at solving the following optimal stopping problem:
\be
V(x,s;c)=\sup_{\tau\in \timset}\ex_{x,s}[\et^{-R_\tau^c}U(\et^{X_{\tau}}) \ind_{\{\tau<\infty\}}],\quad\forall \, (x,s)\in\mc{O}_+ \equiv\{(x',s')\in\R^2: x'\le s'\} ,
 \label{eq:problem1}
 \ee
where $\timset$ is the set of all $\mathbb{F}$-stopping times $\tau$,  $R_t^c$ is the increasing process
\be \label{A}
R_t^c:=rt+\varpi_t^c,\quad\forall \, t\ge0,
\ee
with $r>0$ being the investor's original discount rate and $U(\cdot)$ is either a constant relative risk aversion (CRRA) or a risk-neutral utility function. 
Specifically, by fixing a relative risk aversion coefficient $\rho \in \,[0,1]$, we  define $U(u)$, for $u>0$, by
\be \label{Uu}
U(u) := \begin{cases}
\frac{u^{1-\rho} - 1}{1-\rho} ,\quad&\text{if } \rho \in\,[0,1) ,\\ 
\log (u) ,\quad&\text{if } \rho=1.
\end{cases}
\ee
Note that, when $\rho=0$, by scaling $K\cdot U(\frac{u}{K})=u-K$ we also obtain the revenue from selling the asset with price $u=\et^{X}$ net of the transaction cost $K>0$.

Intuitively, a large anxiety rate $q$ implies a strong penalisation of each unit of time, that the asset price's drawdown breaches the investor's tolerance level $c$. %Such anxiety 
This will naturally demand selling the asset sooner, and in some occasions, this sale may even happen with a loss, if its price does not go up quickly.  
Indeed, our analysis captures this phenomenon by proving that traditional stop-loss type strategies are optimal either when investors have a high tolerance level $c$ for drawdowns, but severe anxiety when these asset price drawdowns occur, see Theorem \ref{thm2} (also Figure \ref{fig:thm5.3}); or when investors have severe anxiety and a low drawdown tolerance, see Theorem \ref{thm3} (also Figures \ref{fig:thm5.7}). 
Furthermore, in the latter class of investors with severe anxiety even about small drawdowns, we prove that trailing stop type strategies also become part of the optimal strategy, see Theorem \ref{thm3} (also Figures \ref{fig:thm5.7}--\ref{fig:notrail}).  
Overall, through the study of problem \eqref{eq:problem1}, this paper manages to answer the following (qualitative and quantitative) question: {\it ``what are the individual tolerance and anxiety characteristics of an investor that may result in an optimal use of trailing stop and stop-loss type strategies?"}. 
Our results on the optimal selling strategy are robust with respect to the choice of utility function. Their qualitative nature remains the same and they only change quantitatively when tuning the risk aversion coefficient $\rho$. 
On the other hand, irrespective of the risk aversion of investors, the absence of severe anxiety always produces take-profit selling strategies (see Section \ref{sec:connection}). Thus, our novel results are mainly driven by anxiety, a risk factor which is not captured by traditional risk aversion.

Drawdown is widely used as a path-dependent risk indicator (see, e.g. \cite{Zhang2018}) and has been often used as a constraint for portfolio optimisation. 
The growth optimal portfolio under exogenous drawdown constraints is explicitly constructed for diffusion models in \cite{Grossman1993, Cvitanic1995}. 
Their strategy entails continuous buying and selling of risky assets, in order to meet the hard constraint on drawdowns. In particular, the portfolio will be 100\% invested in the risk-free asset once the drawdown reaches their pre-specified tolerance level. 
In contrast, our framework only imposes a soft drawdown constraint, in that breaching the investor's drawdown tolerance does not automatically trigger a sale of the risky asset. 
Additionally, our investment is irreversible and indivisible, thus the investor's objective is to find the optimal timing of a sale instead of the continuous rebalancing of the holding position.

The irreversible sale of a real asset which is indivisible, at a time chosen by the investor, is a classical topic in the optimal stopping literature. Optimal timing of an asset sale such that an expected utility is maximised was studied in \cite{Henderson2007} under an exponential utility and in \cite{HendersonHobson2008} under a CRRA utility (see also \cite{HendersonHobsonZeng2018}) in a diffusion framework. 
The problem under a risk neutral utility, which is simply given by the revenue from the sale, namely the value of the asset $\et^{X_\tau}$ net of the transaction cost $K$ at the selling time $\tau$, was studied in \cite{mordecki2002} in a general L\'evy model.\footnote{See also \cite[Example 3.5]{OksendalSulemBook} for a specific geometric L\'evy process and \cite[Example 10.2.2]{Oksendal2003_6} for a geometric Brownian motion} 
The optimal sale of an asset under a utility given by the asset price scaled by its running maximum was studied in \cite{TP_selling09} in a geometric Brownian motion model. 
All aforementioned studies are performed either with a constant discount rate or without discounting.
The study of the above optimal stopping problem with a random (stochastic path-dependent) discount rate and  exponential L\'evy asset price models was developed in \cite{OmegaRZ}. There, a risk neutral utility $U(\et^x) = \et^x-K$ was considered, in a simplified version of the problem formulated in \eqref{eq:problem1}, namely  
\be
\ol{v}(x;y):=\sup_{\tau\in\timset}\ex_x[\et^{-A_\tau^y} U(\et^{X_\tau}) \,  \ind_{\{\tau<\infty\}}],\quad\forall \, x,y\in\R\,, 
\label{eq:oldproblem}
\ee
where instead of the original Omega clock \eqref{eq:omega}, a ``level Omega clock" was used to define $A^y$, given by the occupation time 
\be
A_t^y=rt+q\int_{0}^{t}\ind_{\{X_u<y\}}\diff u \,, \quad \text{ for some fixed risk tolerance level } y\in \R.  
\label{Aold}
\ee 
This ``level Omega clock'' measures the amount of time when $X$ is below this fixed pre-specified level $y$. Above, $\ex_{x}$ is the expectation under $\pr_{x}$, which is the law of $X$ given that $X_0=x\in \mathbb{R}$. 
A European-type equivalent to the problem \eqref{eq:oldproblem} was considered under a risk-neutral utility by \cite{Linetsky99} under a geometric Brownian motion model, which he named a step option.
The American version of this option under the same model and utility and a finite maturity time was recently considered by \cite{Detemple19}.

Part of our analysis builds on the results obtained in Theorems \ref{thm00} and \ref{thm0} for the problem \eqref{eq:oldproblem}, which generalises \cite[Theorems 2.4 and 2.5]{OmegaRZ} to the class of CRRA utility functions. 
To be more precise, we show that the original optimal stopping problem \eqref{eq:problem1} reduces to its simplified version \eqref{eq:oldproblem} for specific ranges of parameter values and certain levels of asset price best performance.   
However, our problem \eqref{eq:problem1} is two-dimensional with state space process $(X,\overline{X})$, which makes the rest of the analysis significantly harder. 
In particular, the optimal exercise boundaries will be proved to be the first entry times of the first component process $X$ in intervals with boundaries, 
which are functions of the second component process $\overline{X}$.
We manage to obtain analytical solutions in all possible cases. 
This is achieved via a guess-and-verify approach, 
as well as the complete characterisation of the solution to a highly non-linear ordinary differential equation.
In mathematical terms, this paper contributes to the literature of explicitly solvable two-dimensional optimal stopping problems in L\'evy models (see \cite{Kyp_Ott12}, \cite{Kyp_Ott14}, \cite{Ott13} among others).

The remaining paper is structured as follows. 
In Section \ref{sec:pre}, we present preliminaries for spectrally negative L\'evy models and standing assumptions. 
We introduce the concepts of mild and severe anxiety in Section \ref{sec:connection} and then present our main results and economic insights in Section \ref{mainres}: Theorem \ref{thm1} (investors with mild anxiety), Theorem \ref{thm2}  and Theorem \ref{thm3} (investors with severe anxiety). 
Figures \ref{fig:H0p}, \ref{fig:thm5.3} and \ref{fig:thm5.7} respectively illustrate numerically the optimal strategies when the log price is given by a compound Poisson jump process plus a Brownian motion with drift. We also simulate our results in Figures \ref{fig:trail} and \ref{fig:notrail}, which demonstrate two possible scenarios faced by an investor with severe anxiety and a small drawdown tolerance. 
The remaining sections are devoted to proving the main results. In particular, 
we bound the value function \eqref{eq:problem1} from above and below in Section \ref{sec:comp} and show its connection with the simpler version of the problem \eqref{eq:oldproblem}.  
In Section \ref{sec:barv}, we fully solve the problem \eqref{eq:oldproblem} and completely characterise all optimal strategies.
Then, we prove the main results of this paper in Section \ref{sec:res}.
Some useful facts on scale functions of spectrally negative L\'evy processes are reviewed in Appendix \ref{app:pre}.

\section{Model and main results}
\label{sec:pre}

We consider a filtered probability space $(\Omega,\mathcal{F},\mathbb{F}=(\mathcal{F}_t)_{t\ge 0},\pr)$ on which we define the logarithm of the asset price (log price) $X=(X_{t})_{t\geq0}$ to be a spectrally negative L\'{e}vy process.
Here $\mathbb{F}$ is the augmented natural filtration of $X$. 
We denote by $(\mu,\sigma^2,\Pi)$ the L\'evy triplet of $X$, and by $\psi$ its Laplace exponent, i.e.
\begin{align*}
\psi(\beta):=&\log\mathbb{E}_0[\et^{\beta X_{1}}]=\mu \beta+\frac{1}{2}\sigma
^{2}\beta^{2}+\int_{-\infty}^{0}\left(  \et^{\beta x}-1-\beta x\ind_{\{x>-1\}}\right)
\Pi(\mathrm{d}x), 
\end{align*}
for every $\beta\in\mathbb{H}^{+}\equiv\{z\in\mathbb{C}: \Re z\ge0\}$. Here,  the L\'{e}vy measure
$\Pi(\mathrm{d}x)$ is supported on $(-\infty,0)$ satisfying $
\int_{-\infty}^{0}(1\wedge x^{2})\Pi(\mathrm{d}x)<\infty$, 
\begin{comment}
It is known that $X$ has paths of bounded variation if and only if
$\int_{]-1,0[}|x|\Pi(\mathrm{d}x)<\infty$ and $\sigma=0$. In this case, we can
rewrite \eqref{decomp} as
\begin{equation}
\psi(\beta)=\gamma\beta+\int_{]-\infty,0[}(\et^{\beta x}-1)\Pi(\mathrm{d}x),
\quad \text{for } \beta\geq0 \quad \text{and} \quad 
\gamma:=\mu+\int_{]-1,0[}|x| \Pi(\mathrm{d}x)>0 \,.
\label{psi bv}%
\end{equation}
\end{comment}
is atom-less and has a (weakly) monotone density $\pi(\cdot)$,  i.e. $\pi(-x)$ is non-increasing in $x>0$ and 
\[\Pi(-\infty,-x)=\int_{-\infty}^{-x}\pi(u)\diff u,\quad\forall x>0.\]
This condition is  weaker than the standard one of the tail measure $\Pi(-\infty,-x)$ either having a completely monotone density or being log-convex. Examples of processes that satisfy the condition include spectrally negative $\alpha$-stable process, spectrally negative CGMY model and spectrally negative hyper-exponential model.

We assume that the discount rate $r>\psi(1)$, which is equivalent to the discounted asset 
price $(\et^{-rt+X_t})_{t\ge0}$ being a super-martingale. 
For any given $r\geq0$, the equation $\psi(\beta)=r$ has at least
one positive solution, and we denote the largest one by $\Phi(r)$. 
Notice that $r>\psi(1)$ implies that $\Phi(r)>1$.
Finally, the $r$-scale function $W^{(r)}:\mathbb{R}\mapsto[0,\infty)$ is a function vanishing 
on $(-\infty,0)$, continuous on $[0,\infty)$, with a Laplace transform given by
\ben
\int_0^\infty \et^{-\beta x}W^{(r)}(x)\mathrm{d}x 
= \frac{1}{\psi(\beta)-r},\quad \text{for } \beta>\Phi(r) \,. 
\een
We assume that $W^{(r)}(\cdot)\in C^{2}(0,\infty)$ for all $r\geq0$, which is guaranteed if 
$\sigma>0$ (see e.g. \cite[Theorem 3.10]{Kuznetsov_2011}).\footnote{However, $\sigma>0$ is not a necessary condition for $W^{(r)}(\cdot)\in C^2(0,\infty)$. For instance, a spectrally negative $\alpha$-stable process with $\alpha\in(1,2)$  satisfies this condition without a Gaussian component.} The $r$-scale function $W^{(r)}(\cdot)$ is closely related to the first passage times of the spectrally negative L\'{e}vy process $X$, which are defined by\,\footnote{\label{note1}Notice also that the assumption $r>\psi(1)$ implies  
$$0\le \lim_{z\to\infty}\ex_x[\et^{-A_{T_z^+}^y}U(X_{T_z^+})\ind_{\{T_z^+<\infty\}}]\le\lim_{z\to\infty}\ex_x[\et^{-rT_z^+}U(X_{T_z^+})\ind_{\{T_z^+<\infty\}}]=\lim_{z\to\infty}U(z)\et^{-\Phi(r)z}=0.$$}
\ben
T_{x}^{\pm}:=\inf\{  t\geq0:X_{t} \gtrless x\},\quad  x\in\mathbb{R}.
\een

In Appendix \ref{app:pre} we  list a handful of useful properties and identities of $r$-scale functions.

\subsection{Two degrees of anxiety}
\label{sec:connection}

As seen in \cite{OmegaRZ}, the optimal selling strategy, for problem \eqref{eq:oldproblem} with a risk neutral utility, is largely affected by the investor's own anxiety rate $q$.  
To extend the arguments there to a general CRRA utility and eventually to our objective \eqref{eq:problem1}, we introduce the concept of  {\it mild anxiety} and {\it severe anxiety}, after conducting some preliminary analysis of problem \eqref{eq:oldproblem}.  

We focus on the case of $\rho\in[0,1)$, because the corner case of logarithmic utility can be obtained from the isoelastic utility in the limit $\rho \uparrow 1$.  

Our preliminary approach is inspired by the optimality of threshold type strategies  \cite{Long18}. We first consider the benchmark case of no anxiety, i.e. $q=0$. In this case,  we always discount at rate $r>0$. Given that 
\be
z\mapsto U(\et^z)-\frac{1}{\Phi(r)}\frac{\diff}{\diff z}U(\et^z)=\et^{(1-\rho)z}\bigg(\frac{1}{1-\rho}-\frac{1}{\Phi(r)}\bigg)-\frac{1}{1-\rho}\label{eq:hzz}\ee
is strictly increasing over $\R$, we know from \cite[Theorem 2.2]{Long18} (see also footnote \ref{note1}) that problem \eqref{eq:oldproblem} is solved by a take-profit (up-crossing) selling strategy when the log price reaches the target
$\frac{1}{1-\rho}\log(\frac{\Phi(r)}{\Phi(r)-1+\rho})>0.$

With anxiety, i.e. $q>0$, we consider the function $g(\cdot)$ given by 
\begin{align}
g(x)=\et^{(1-\rho)x} \bigg(\frac{1}{1-\rho} - \frac{1}{\Lambda(x)}\bigg)
\quad \text{with} \quad 
\Lambda(x):= 
\frac{\diff}{\diff x}\log\mc{I}^{(r,q)}(x), 
\label{Lg}
\end{align}
where $\mc{I}^{(r,q)}(\cdot)$ is defined 
by 
\be \label{Irq}
\mc{I}^{(r,q)}(x):=\int_{0}^{\infty}\et^{-\prq u}W^{(r)}(u+x)\diff u,\quad\forall \,x\in\R \,.
\ee
The function $g(\cdot)$ is continuous everywhere with one possible discontinuity at 0. 
It is known from \cite[Lemma 4.2]{OmegaRZ} that $\Lambda(\cdot)$ is strictly decreasing over $\R_+$, with limits $\Lambda(0+)\le \Phi(r+q)$ and $\Lambda(\infty)=\Phi(r)$. 
Moreover, following similar analysis to \cite{OmegaRZ}, we know that $g(\cdot)$ is strictly increasing over  $(-\infty,0)$ with a lower limit $g(-\infty)=0$, and is ultimately increasing over $[u,\infty)$ for $u\ge0$ sufficiently large, with an upper limit $g(\infty)=\infty$.  Thus,  one can unambiguously define (the largest local minimum of $g$) 
\ben 
\ol{u}:=\inf\{u\in\R: g(\cdot) \text{ is non-decreasing over }[u,\infty)\} \,.
\een
Notably, constant $\ol{u}$ is either non-negative, or $-\infty$. 

In case $\ol{u}=-\infty$, the function $g(\cdot)$ is non-decreasing over $\R$,  so for any fixed $y\in\R$, the mapping
\be \label{eq:g}
z\mapsto  U(\et^z)-\frac{1}{\Lambda(z-y)}\frac{\diff}{\diff z}U(\et^z)=\et^{(1-\rho)y}g(z-y)-\frac{1}{1-\rho},\quad\forall z\in\R
\ee 
is also non-decreasing. By \cite[Theorem 2.2]{Long18} and Lemma \ref{lem:I} (see also footnote \ref{note1}), we immediately know that \eqref{eq:oldproblem} is also solved by a take-profit selling strategy, regardless of the risk tolerance level $y$. In this case, the selling target log price is given by\,\footnote{The definition of $z^\star(y)$ holds for any value of $\ol{u}$.}
\be \label{z*} 
z^\star(y):=y+\inf\left\{u>\ol{u}: g(u) > \et^{-(1-\rho)y}/(1-\rho) \right\}\,.
\ee 
Equivalently, $z^\star(y)$ is the largest root over $(y+\ol{u},\infty)$, which is simply $\R$ when $\bar u = - \infty$, to equation
\be
g(z-y)=\frac{1}{1-\rho}\et^{-(1-\rho)y}.\label{eq:zstarg}
\ee
Because of the qualitative similarity of this type of optimal strategy with that under no anxiety (i.e. $q=0$), we henceforth refer to the case of $\ol{u}=-\infty$ as the case of {\it mild anxiety}. 

The remaining case, when $\ol{u}\ge0$, is referred to as the case of {\it severe anxiety}.
In this case, contrary to the previous mild one, an investor %In contrast to the case of mild anxiety, an investor with severe anxiety 
may choose an additional stop-loss type strategy to ``cut the loss'' depending on the risk tolerance level $y$, which is a phenomenon already documented in \cite{OmegaRZ} under a risk neutral utility\,\footnote{We shall see that such distinction still exists for our generalised problem \eqref{eq:oldproblem} (see Section \ref{sec:barv}) and our main objective \eqref{eq:problem1} (see main results in Section \ref{mainres} below).}. 
Furthermore,  the representation of candidate threshold $z^\star(y)$ as the largest root to \eqref{eq:zstarg} does not always hold under severe anxiety. Specifically, only if 
\be \label{ybar}
y\le \bar{y}:=-\frac{1}{1-\rho}\log\big((1-\rho)g(\ol{u})\big),
\ee
can we identify $z^\star(y)$ as the largest root over $(y+\ol{u},\infty)$ to \eqref{eq:zstarg}. If $y>\ol{y}$, $z^\star(y)$ of \eqref{z*} is simply equal to $y+\ol{u}$.

\begin{figure}
\centering
\includegraphics[width=3in]{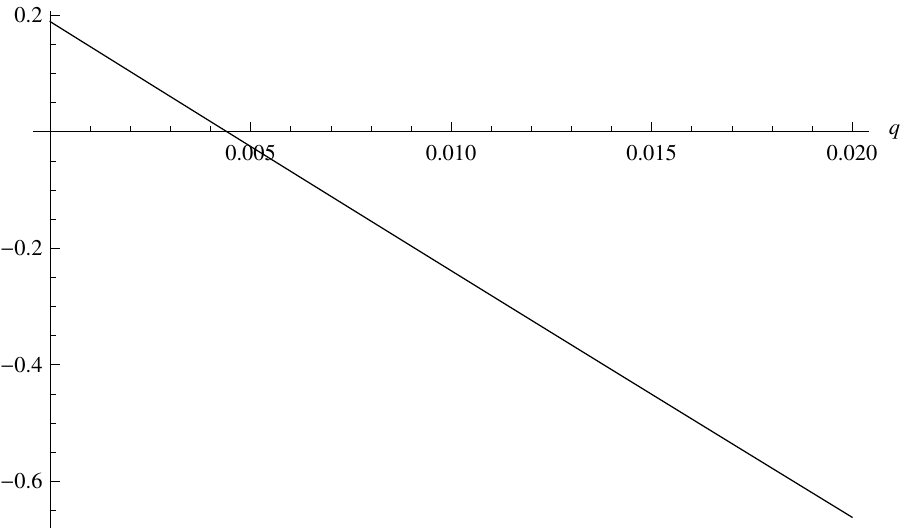}
\caption{Plot of $H^\star$ as a function of $q$. 
Here, we consider a risk neutral investor (i.e. $\rho=0$) with a discount rate $r=0.18$ and the Laplace exponent used is $\psi(\beta)=0.18\beta+0.02\beta^2-0.25(\frac{\beta}{\beta+4})$ (so the jump distribution is exponential).}
\label{fig:figH}
\end{figure}

We close this subsection by providing a convenient criterion that distinguishes severe from mild anxiety:
when the tail jump measure of the L\'evy process, denoted by $\bar{\Pi}(x):=\Pi(-\infty,-x)$ (for $x>0$), either has a completely monotone density or is log-convex, then following similar arguments as in 
\cite[Lemma 2.6]{OmegaRZ}, one can show that 
\be
\label{eq:explicit}
\ol{u}=-\infty \text{ holds } \quad \Leftrightarrow \quad 
H^\star := (\prq-1+\rho)(\prq-qW^{(r)}(0))-qW^{(r)\prime}(0+)\ge0 \,.
\ee
Figure \ref{fig:figH} plots $H^\star$ as a function of $q$, illustrating the relationship between mild anxiety ($\ol{u}=-\infty$) and small $q$, as well as severe anxiety ($\ol{u}\geq 0$) and large $q$.

\subsection{Main results} \label{mainres}

In this section we present our main result, the value function and the optimal selling strategy for problem  \eqref{eq:problem1}, when the investor has either mild  (i.e. $\ol{u}=-\infty$) or severe (i.e. $\ol{u}\ge0$) anxiety.   

\vspace{3pt}
To begin, we note that the optimal stopping of reward $U(\et^X)$ with a constant discounting rate $r+q$, is solved by a take-profit selling strategy with target log price 
\be \label{b_}
\ul{b} = \frac{1}{1-\rho} \, \log \Big( \frac{\Phi(r+q)}{\Phi(r+q)-1+\rho} \Big) > 0 \,,
\ee  
and the value function is given by
\begin{align} \label{eq:underlinev} 
\underline{v}(x)
&:=\sup_{\tau\in\timset}\ex_x[\et^{-(r+q)\tau} U(\et^{X_{\tau}}) \ind_{\{\tau<\infty\}}] 
=\ind_{\{x\le\underline{b}\}}\et^{\Phi(r+q)(x-\underline{b})}U(\et^{\, \underline{b}}) + \ind_{\{x>\underline{b}\}} U(\et^x) \,.
\end{align} 
Function $\ul{v}(\cdot)$ is smooth everywhere off the set $\{\ul{b}\}$, and is continuously differentiable over $\R$. These results  follow from \cite[Theorem 2.2]{Long18} together with the increasing property of mapping \eqref{eq:hzz} as we replace discount rate $r$ by $r+q$.

Define also the log price threshold 
\be \label{z_c}
z_c := \frac{1}{1-\rho}\log\bigg(\frac{\Lambda(c)}{\Lambda(c)-1+\rho}\bigg). 
\ee
By the monotonicity of $\Lambda(\cdot)$ over $\R_+$, we know that $z_c>\ul{b}$.

\subsubsection{Investors with mild anxiety}
\label{sec:mild}

In the first result of this section, %concerned with investors with mild anxiety, 
we present the properties of two types of take-profit sale targets $\et^{z^\star(\cdot)}$ from \eqref{z*} or $\et^{z_c}$ from \eqref{z_c}, which can be attained either before or after the asset price improves its best performance $\et^{s}$, respectively. 

\begin{lem}\label{lem21}
If $\ol{u}=-\infty$, then
\begin{enumerate}
\item[(i).] the function $z^\star(\cdot)$ of \eqref{z*} is continuous and strictly decreasing over $(-\infty,\ul{b}]$, and $z^\star(y)\equiv \ul{b}$ for all $y\ge\ul{b}$;
\item[(ii).] the log price $z_c$ from \eqref{z_c} satisfies $z^\star(z_c-c)=z_c$ and $z_c<\ul{b}+c$.
\end{enumerate}
\end{lem}
\begin{figure}
\centering
\includegraphics[width=4in]{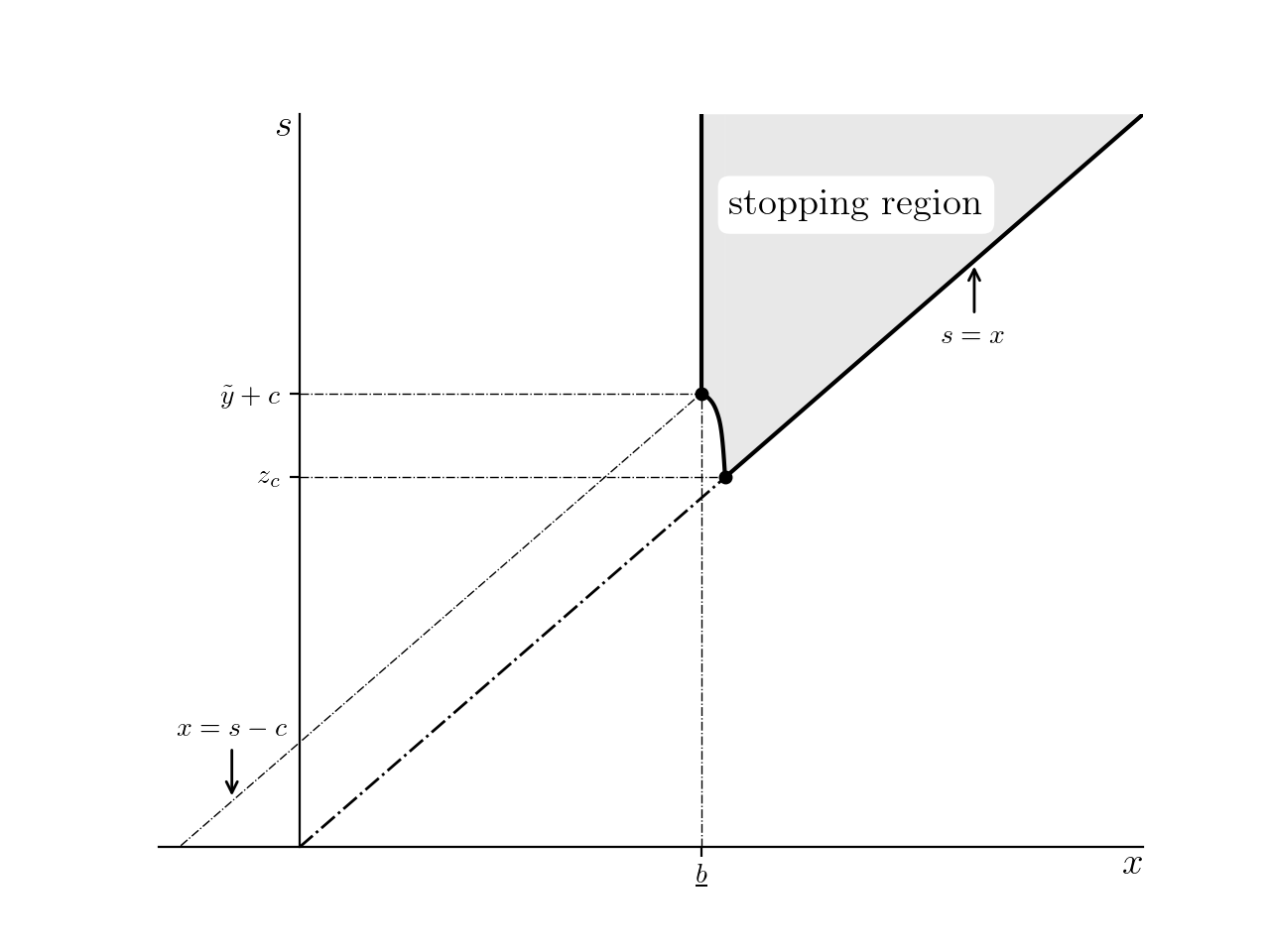}
\caption{Illustration of the optimal stopping region under mild anxiety 
($\ol{u} = -\infty$).
Model parameters:
 $c=0.3568, r=0.18, q=0.003$ and Laplace exponent used: $\psi(\beta)=0.18\beta+0.02\beta^2-0.25(\frac{\beta}{\beta+4})$. The investor solves a risk-neutral sale with transaction cost  $K=10$. Here, we have $H^\star=0.0603, \underline{b}=4.1903$ and $z_c=4.2616$. The axes origin is at $(3,3)$ in the figure. }
\label{fig:H0p}
\end{figure}

The optimality of the above selling strategies is given in the following result and is proved in Section \ref{sec:51} via the use of variational inequalities (see the beginning of Section \ref{sec:res}) and the results obtained in the subsequent Sections \ref{sec:comp} and \ref{sec:barvm}. 

\begin{thm}\label{thm1}
For an investor with mild anxiety (i.e. $\ol{u}=-\infty$), 
the value function for problem \eqref{eq:problem1} is given by
\be
V(x,s;c)=\begin{dcases}\et^{-\Lambda(c)(z_c-s)}
\frac{\mc{I}^{(r,q)}(x-s+c)}{\mc{I}^{(r,q)}(c)} \, U(\et^{z_c}) 
,&\text{if } s< z_c;\\
\ind_{\{x<z^\star(s-c)\}}\frac{\mc{I}^{(r,q)}(x-s+c)}{\mc{I}^{(r,q)}(z^\star(s-c)-s+c)} \, U(\et^{z^\star(s-c)}) +\ind_{\{x\ge z^\star(s-c)\}} U(\et^x),&\text{if }  s\ge z_c. 
\end{dcases}\label{eq:thm1}
\ee
The optimal selling region (see Figure \ref{fig:H0p} for an illustration) is 
$$
\mc{S}_c =
\{(x,s)\in\mc{O}_+ \;:\; x \ge z^\star(s-c) \text{ and } \; s\ge z_c \}.
$$
\end{thm} 

Theorem \ref{thm1} asserts that an investor with mild anxiety does not behave so differently from an investor with no anxiety ($q=0$) when facing problem \eqref{eq:problem1}. The optimal strategy is always given by a take-profit sale. 
However, observe that the two regions of the value function \eqref{eq:thm1}, when $s \lessgtr z_c$, are not communicating. The optimal selling target thus depends explicitly on the past, through the asset's historical best performance at the price $\et^s$. To be more precise:

\begin{enumerate}
\item[(i)] When the starting maximum log price $\ol{X}_0=s$ is lower than the target $z_c$, the investor should hold onto the asset and sell it once its log price reaches the threshold $z_c$;

\item[(ii)]  When the starting maximum log price $\ol{X}_0=s$ is already higher than the target $z_c$, then the investor is less patient about adverse movements of the asset price, hence lowers the optimal selling target to the log price $z^\star(s-c)$.
\end{enumerate}

\subsubsection{Investors with severe anxiety}

In this case, %of {\it severe anxiety}, 
i.e. when $\ol{u}\ge0$, the structure of the optimal selling strategy for problem \eqref{eq:problem1} changes as the investor's tolerance level $c$ varies. 
In order to define the critical regions of $c$-values, we firstly need to specify two values $\hat y$ and $\tilde y$ that are closely related to variational inequalities and martingale methods associated to the optimality of take-profit selling strategies.   
To formalise the following results, consider the function 
\be \label{eq:57}
\chi(x):=
\frac{r}{1-\rho} - \frac{r-\psi(1-\rho)}{1-\rho} \et^{(1-\rho)x} +\int_{-\infty}^{\underline{b}-x} \left(\ul{v}(x+w)-U(\et^{x+w}) \right) \Pi(\diff w),\quad \forall\;x\ge\underline{b}. 
\ee
Following similar analysis to \cite{OmegaRZ}, one can show that $\chi(\cdot)$ is continuous and strictly decreasing over $[\underline{b},\infty)$, and satisfies $\chi(\infty)=-\infty$. Thus, we can define 
\be \label{Lv_1}
\hat{y} := \inf \left\{ y \ge \ul{b} \,:\, 
\chi(y)
\leq 0 \,\right\}.
\ee 
Moreover, we define another critical log price threshold $\tilde{y}$, given by 
\be \label{tily}
\tilde{y} := \inf \left\{ y \le\ol{y} 
\,:\, \sup_{x<y+\ol{u}} \frac{U(\et^x)}{\mc{I}^{(r,q)}(x-y)} = \frac{U(\et^{z^\star(y)})}{\mc{I}^{(r,q)}(z^\star(y)-y)}\right\}.
\ee 

In the following lemma, %concerned with investors with severe anxiety, 
we present the properties of three types of take-profit sale targets. 
On the one hand, $\et^{z_c}$ from \eqref{z_c} that can be attained after the asset price improves its best performance $\et^{s}$. 
On the other hand, $\et^{z^\star(\cdot)}$ from \eqref{z*} or the novel $b^\star(\cdot)$, which is also associated with a stop-loss type sale target $a^\star(\cdot)$, that can be attained before the asset price improves its best performance $\et^{s}$. 

\begin{lem}\label{lem23}
If $\ol{u}\ge0$, then
\begin{enumerate}
\item[(i).] the function $z^\star(\cdot)$ of \eqref{z*} is continuous and strictly decreasing over $(-\infty, \tilde{y}]$;
\item[(ii).] we have $\ul{b}<\tilde{y}<\ol{y}$ and $\tilde{y}<\hat{y}$;
\item[(iii).] by defining the positive\,\footnote{Positivity of $\tilde{c}$ is due to $\tilde{c}>z^\star(\ol{y})-\ol{y}\ge\ol{u}\ge0$.}
\hspace{-2pt}constant 
\be
\tilde{c}:=z^\star(\tilde{y})- \tilde{y}, 
\label{tilc}
\ee 
we have for all $c\ge\tilde{c}$ that $z^\star(z_c-c)=z_c$.
\item[(iv).] for any $y\in[\tilde{y}, \hat{y})$, let $a^\star(y), b^\star(y)$ be defined as
\be
\label{a*b*}
a^\star(y)=\inf\{a\in[\underline{b},y): \mc{N}^a\neq\emptyset\} 
\;\; \text{ and } \;\;
b^\star(y)=\inf\mc{N}^{a^\star(y)},\ee
where $\mc{N}^a:=\{x\in \,(y,z^\star(\tilde{y})]: \Delta(x,a;y)\le 0\}$ is the negative set of the function indexed by $a\in[\ul{b}, y)$:
\begin{equation} \label{D}
x\mapsto\Delta(x,a;y)=\int_{a}^yW^{(r,q)}(x,w;y)\cdot[q\,\underline{v}(w)-\chi(w)]\diff w-\int_{y}^{x\vee y}W^{(r)}(x-w)\cdot\chi(w)\diff w \,,
\end{equation} 
with $W^{(r,q)}(x,a;y)$ given in \eqref{eq:bigW1}. 
Then $a^\star(\cdot)$ and $b^\star(\cdot)$  are respectively strictly increasing/decreasing continuous functions on $[\tilde{y},\hat{y})$, and satisfy
\be
a^\star(\tilde{y})=\ul{b}, \quad b^\star(\tilde{y}) = z^\star(\tilde{y})\equiv z_{\tilde{c}},\quad  \lim_{y\uparrow \hat{y}}a^\star(y)=\lim_{y\uparrow \hat{y}}b^\star(y)=\hat{y}.\label{eq:blimit}\ee
\end{enumerate}
\end{lem}

We differentiate investors with severe anxiety into two types: the ones with a {\it high} and the others with a {\it low tolerance} for asset price drawdowns, based on whether the investor's drawdown tolerance level $c\ge\tilde{c}$ or not. Here $\tilde{c}$ is the constant defined in \eqref{tilc}.

The optimality of the above selling strategies for the former class of investors %with severe anxiety and high drawdown tolerance 
is given in the following result. This is proved in Section \ref{highc} via the use of variational inequalities (see the beginning of Section \ref{sec:res}) and the results obtained in the subsequent Sections \ref{sec:comp} and \ref{sec:barvs}. 

\begin{thm}\label{thm2}
For an investor with severe anxiety (i.e. $\ul{u}\ge0$) and a high drawdown tolerance $c\ge\tilde{c}$, the value function for problem \eqref{eq:problem1} is given by 
\be
V(x,s;c)=\begin{dcases}
\et^{-\Lambda(c)(z_c-s)} \frac{\mc{I}^{(r,q)}(x-s+c)}{\mc{I}^{(r,q)}(c)} \, U(\et^{z_c}), 
&\text{if } s< z_c;\\
\ind_{\{x<z^\star(s-c)\}}\frac{\mc{I}^{(r,q)}(x-s+c)}{\mc{I}^{(r,q)}(z^\star(s-c)-s+c)} \, U(\et^{z^\star(s-c)}) +\ind_{\{x\ge z^\star(s-c)\}} U(\et^x), 
&\text{if } z_c \le s < \tilde{y} +c; \\
\underline{v}(x) + \ind_{\{a^\star(s-c)<x<b^\star(s-c)\}} \, \Delta(x,a^\star(s-c);s-c),
&\text{if } \tilde{y}+c \leq s < \hat{y}+c; \vspace{4.5pt}\\
\underline{v}(x),
&\text{if } s \geq \hat{y}+c. 
\end{dcases}\label{eq:thm2}
\ee
The optimal selling region (see Figure \ref{fig:thm5.3} for an illustration) is $\mc{S}_c = \mc{S}_c^1\cup\mc{S}_c^2\cup\mc{S}_c^3$, where
\begin{align*}
\mc{S}_c^1&=
\{(x,s)\in\mc{O}_+ \;:\; x\ge z^\star(s-c) \,\text{ and } \; z_c\le s<\tilde{y}+c \},\\
\mc{S}_c^2&=\{(x,s)\in\mc{O}_+ \;:\;  \ul{b}\le x\le a^\star(s-c) \text{ or } x\ge b^\star(s-c), \,\text{ and } \; \tilde{y}+c\le s< \hat{y}+c \},\\
\mc{S}_c^3&=\{(x,s)\in\mc{O}_+ \;:\; x \ge \ul{b} \,\text{ and } \;  s \geq \hat{y}+c \}.  
\end{align*}
\end{thm}

\begin{figure}
\centering
\includegraphics[width=4in]{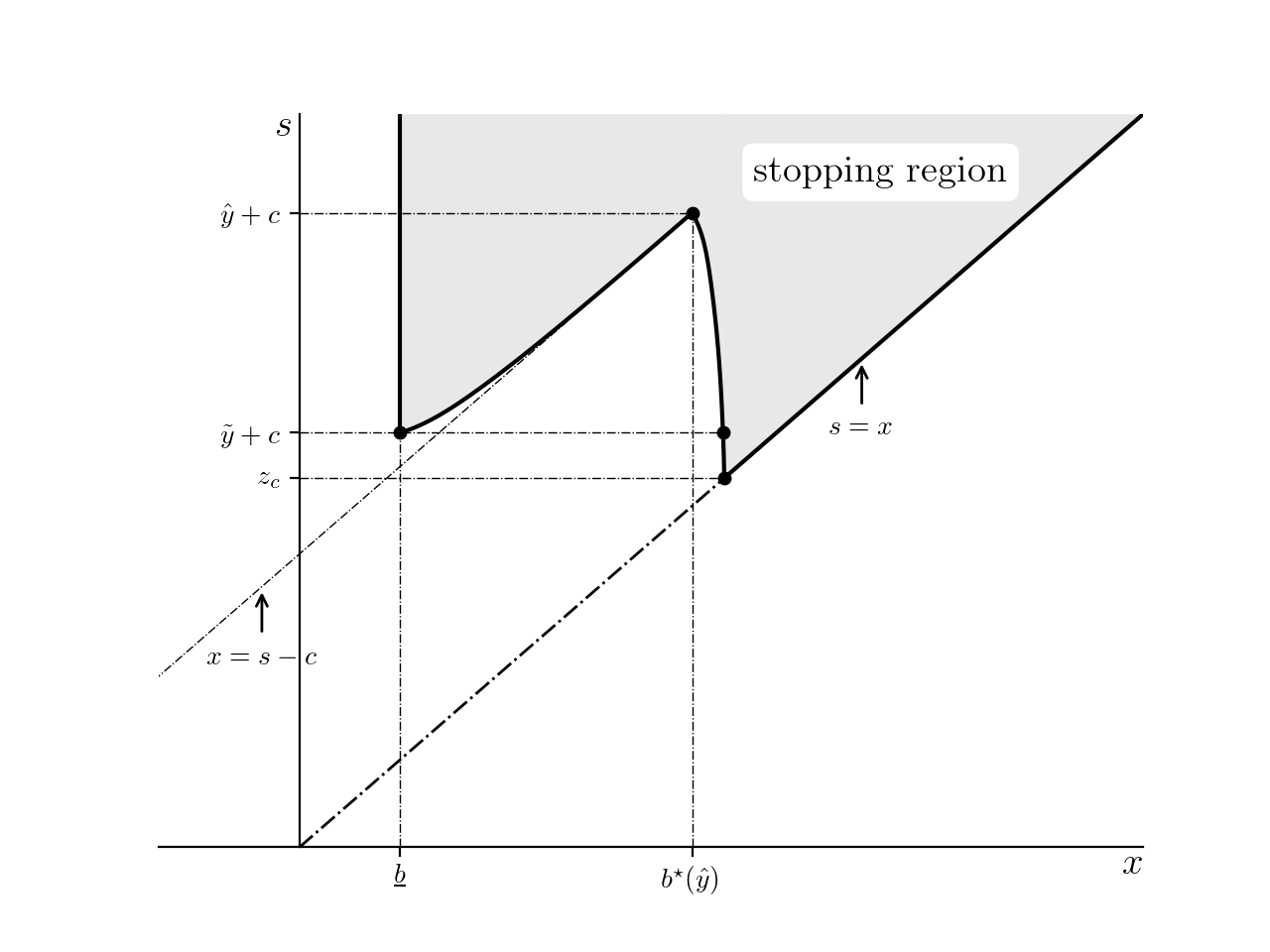}
\caption{Illustration of the optimal stopping region under severe anxiety 
($\ol{u} > 0$) 
and high drawdown tolerance ($c \geq \tilde{c}$). Model parameters: 
$c=1.8, r=0.18, q=1$ and Laplace exponent used: $\psi(\beta)=0.18\beta+0.02\beta^2-0.25(\frac{\beta}{\beta+4})$. The investor solves a risk-neutral sale with transaction cost  $K=10$. Here we have $H^\star=-31.9618, \underline{b}=2.5375, \tilde{c}=1.5143, \tilde{y}=2.7452, \hat{y}=4.0946$ and $z_c=4.2664$. The axes origin is at $(2,2)$ in the figure. }
\label{fig:thm5.3}
\end{figure}

The difference in the degree of investor's anxiety in Theorems \ref{thm1} and \ref{thm2} results in the optimal strategies having different structures when $s\ge z_c$.

\begin{enumerate}
\item[(i)] If the initial maximum log price $\ol{X}_0=s$ is strictly lower than $z_c$, then it is optimal to hold onto the asset and sell it once its log price reaches threshold $z_c$. In this case, the investor essentially behaves in the same way as an investor with mild anxiety.

\item[(ii)] If the best performance of the asset is higher than $z_c$, the investor with severe anxiety may need to complement the take-profit sale with an additional stop-loss type sale to ``cut the loss''. 
\end{enumerate}

\vspace{3pt}
Finally, we treat the type of investors with {\it severe anxiety}, who are concerned also about asset price drawdowns with size smaller than $\tilde{c}$. By fixing a $c\in(0,\tilde{c})$, and using the continuity, monotonicity of $b^\star(\cdot)$ over $[\tilde{y}, \hat{y})$, and its limits in \eqref{eq:blimit}, we may define a threshold $y_c$ as the unique root over $(\tilde{y}, \hat{y})$ to equation
 \be \label{z_cn}
 b^\star(y_c) - y_c =c. 
\ee
By construction we know that $\hat{y}< b^\star(y_c)<z_{\tilde{c}}\equiv b^\star(\tilde{y})$.
We shall see that the optimal selling strategy involves a combination of the take-profit sale at the target log price $b^\star(y_c)$, in conjunction with a trailing-stop type sale, with the barrier given by the solution to a first order non-linear ordinary differential equation (ODE). We provide the complete characterisation of this solution in the following lemma.
\begin{lem}\label{lem:as}
Let $\Delta(\cdot,a;y)$ and $a^\star(\cdot)$ be defined by \eqref{a*b*}--\eqref{D} and $f(x) := \int_{-\infty}^{\underline{b}-x}(\ul{v}(x+w)-U(\et^{x+w}))\Pi(\diff w)$.
Then, there exists a unique solution $a(\cdot)$ to the first order non-linear ODE
\be
\label{eq:ODE}
\begin{dcases}
a^\prime(s)=\dfrac{qW^{(r)}(c)}{W^{(r,q)}(s,a(s);s-c)}\, \dfrac{(1-\rho) \, \big(\underline{v}(s-c) + \Delta(s-c,a(s);s-c)\big)}{(r+q-\psi(1-\rho)) \et^{(1-\rho)\, a(s)} - (r+q) - (1-\rho) \, f(a(s))},\quad \forall \,s\le b^\star(y_c),\\
a(b^\star(y_c))=a^\star(y_c),
\end{dcases}
\ee
which can be extended smoothly for $s\le b^\star(y_c)$ as long as $a(s)\ge\underline{b}$.
Moreover, there exists a unique $s_c\in\,(\underline{b}+c,b^\star(y_c))$ such that $a(s_c)=\underline{b}$ and we have 
\[
a'(s)>0, \quad a(s)<s-c,\quad \Delta(x,a(s);s-c)>0,\quad\forall\, (x,s)\in \mc{O}_+\text{ s.t. }a(s)<x\le s-c \text{ and } s_c\le s\le b^\star(y_c).
\]
\end{lem} 

The optimality of the selling strategies presented in Lemma \ref{lem23} together with the aforementioned trailing stop type sale target, is given below for %investors with severe anxiety and low drawdown tolerance. 
this class of investors. 
This is proved in Section \ref{highc} via the use of variational inequalities (see the beginning of Section \ref{sec:res}) and the results obtained in the subsequent Sections \ref{sec:comp} and \ref{sec:barvs}. 

\begin{figure}
\centering
\includegraphics[width=4in]{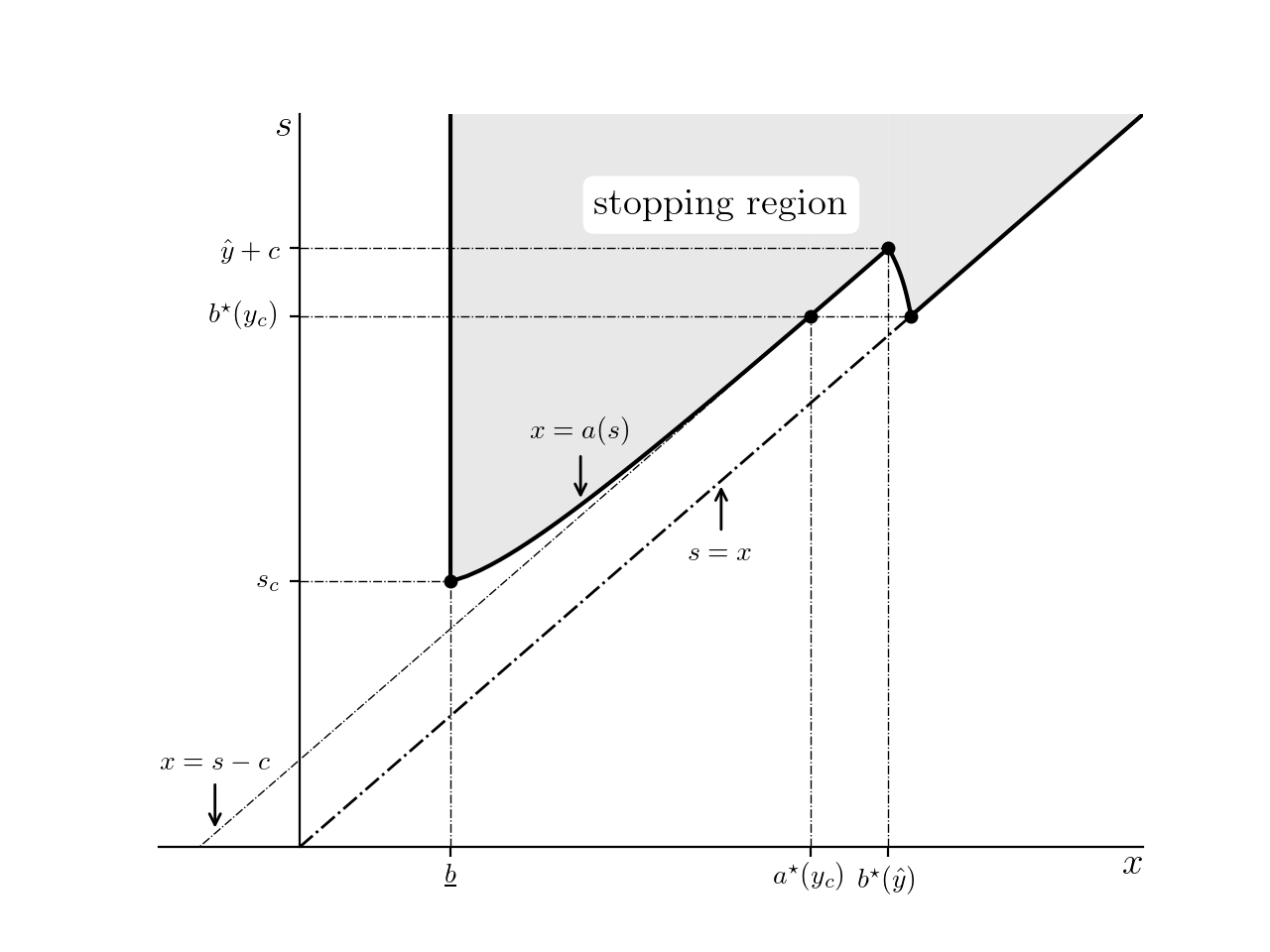}
\caption{Illustration of the optimal stopping region under severe anxiety 
($\ol{u} > 0$)
and low drawdown tolerance ($c<\tilde{c}$). Model parameters: 
$c=0.3568, r=0.18, q=1$ and Laplace exponent used: $\psi(\beta)=0.18\beta+0.02\beta^2-0.25(\frac{\beta}{\beta+4})$. The investor solves a risk-neutral sale with transaction cost  $K=10$. Here we have $H^\star=-31.9618, \underline{b}=2.5375, \tilde{c}=1.5143, \hat{y}=4.0946, s_c=3.0877$ and 
$b^\star(y_c)=4.1748$. The axes origin is at $(2,2)$ in the figure. }
\label{fig:thm5.7}
\end{figure}

\begin{thm}\label{thm3}
For an investor with severe anxiety (i.e. $\ol{u}\ge0$) and a low drawdown tolerance $c<\tilde{c}$, the value function for problem \eqref{eq:problem1} is given by
\be
V(x,s;c)=\begin{dcases}
\et^{-\Lambda(c)(s_c-s)}
\frac{\mc{I}^{(r,q)}(x-s+c)}{\mc{I}^{(r,q)}(c)} V(s_c,s_c;c), 
&\text{if } s<s_c;\\ \vspace{0.08cm}
\underline{v}(x) + \Delta(x,a(s);s-c) \ind_{\{x\ge a(s)\}}, 
&\text{if } s_c\le s < b^\star(y_c); \\ 
\underline{v}(x) + \ind_{\{a^\star(s-c)< x<b^\star(s-c)\}} \Delta(x,a^\star(s-c);s-c),
&\text{if } b^\star(y_c) \leq s < \hat{y}+c; \vspace{2.5pt}\\
\underline{v}(x),
&\text{if } s \geq \hat{y}+c.
\end{dcases}\label{eq:thm3}
\ee
The optimal selling region (see Figure \ref{fig:thm5.7} for an illustration) is $\mc{S}_c=\mc{S}_c^1\cup\mc{S}_c^2\cup\mc{S}_c^3$, where
\begin{align*}
\mc{S}_c^1=&\{(x,s)\in\mc{O}_+ \;:\; \ul{b}\le x \le  a(s) \,\text{ and } \; s_c \leq s < b^\star(y_c)\},\\
\mc{S}_c^2=&\{(x,s)\in\mc{O}_+ \;:\; \ul{b}\le x\le  a^\star(s-c) \text{ or } x\ge b^\star(s-c), \,\text{ and } \; b^\star(y_c)\leq s < \hat{y}+c \},\\
\mc{S}_c^3=&\{(x,s)\in\mc{O}_+ \;:\; x \ge\ul{b} \,\text{ and } \; s \geq \hat{y}+c \}.
\end{align*}
\end{thm}

\begin{figure}
\centering
\includegraphics[width=4in]{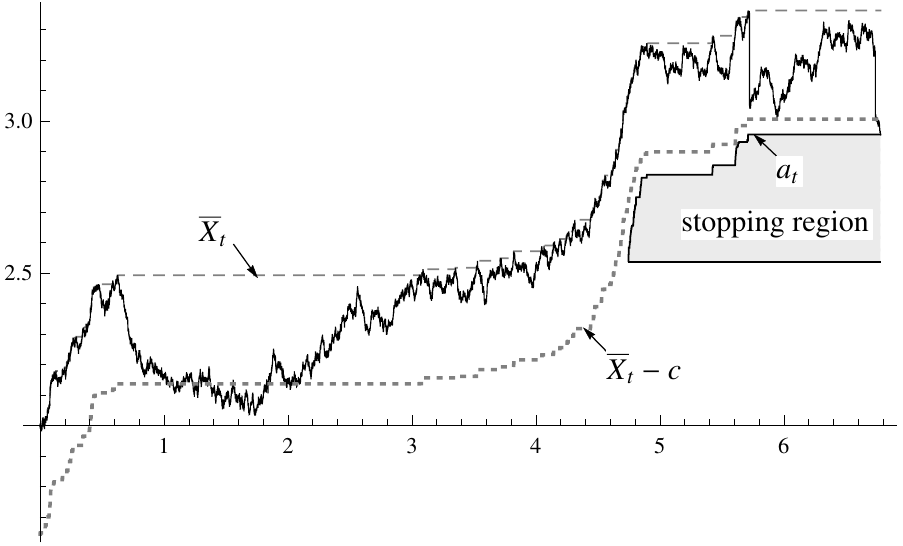}
\caption{A simulated sample path under parameters: $c=0.3568$, $r=0.18, q=1$,  Laplace exponent $\psi(\beta)=0.18\beta+0.02\beta^2-0.25(\frac{\beta}{\beta+4})$,  and initial value $X_0=\ol{X}_0=2$. The investor solves a risk-neutral sale with transaction cost  $K=10$. In the figure, $a_t=a(\ol{X}_t)$ is the trailing stop threshold, which is set up once $\ol{X}$ reaches $s_c=3.0877$. This particular path fails to reach the upper threshold $b^\star(y_c)=4.1748$ (not shown) before activating the trailing stop at value $2.949$ and $t=6.7795$.}
\label{fig:trail}
\end{figure}

\begin{figure} 
\centering
\includegraphics[width=4in]{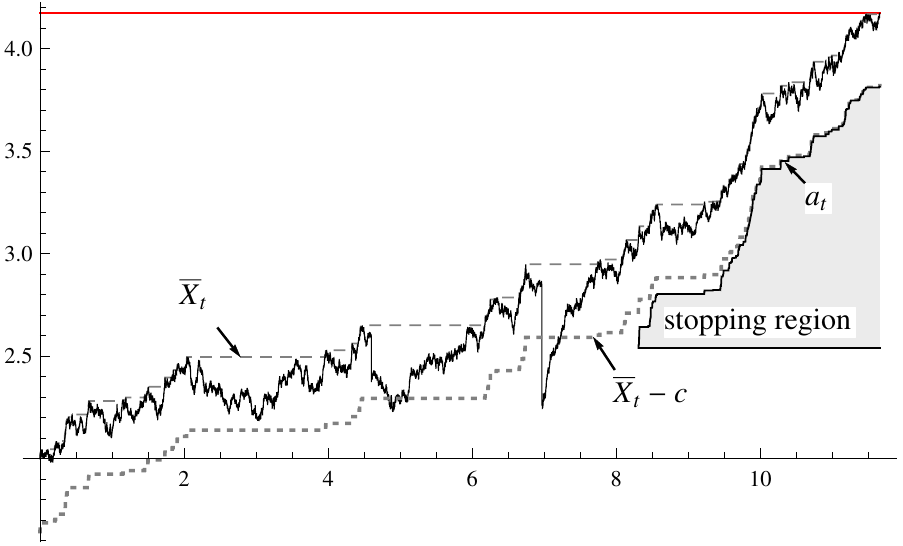}
\caption{A simulated sample path under parameters: $c=0.3568$, $r=0.18, q=1$,  Laplace exponent $\psi(\beta)=0.18\beta+0.02\beta^2-0.25(\frac{\beta}{\beta+4})$,  and initial value $X_0=\ol{X}_0=2$. The investor solves a risk-neutral sale with transaction cost  $K=10$. In the figure, $a_t=a(\ol{X}_t)$ is the trailing stop threshold, which is set up once $\ol{X}$ reaches $s_c=3.0877$. This particular path succeeds in reaching the upper threshold $b^\star(y_c) = 4.1748$ (shown in red line) before activating the trailing stop. The hitting occurs at time $t=11.656$.}
\label{fig:notrail}
\end{figure}

Contrary to Theorems \ref{thm1} and \ref{thm2}, we observe that the regions of the value function communicate in Theorem \ref{thm3}. 
If an investor with severe anxiety has a low tolerance for drawdowns, then the optimal selling strategy may involve some holding period of no trade, and some period when a fixed take-profit target is set up together with a sequence of protective, adaptive stop-loss orders.
To be more precise:

\begin{enumerate}
\item[(i)] When the starting maximum log price $s$ is lower than $s_c$, the investor should hold onto the asset until its log price reaches $s_c$ (no trade region). 

\item[(ii)] If the maximum log price $\bar{X}$ is at least $s_c$, but lower than $b^\star(y_c)$, the investor should set up a take-profit target log price $b^\star(y_c)$, while also consider to (optimally) sell the asset when the log price $X$ jumps down to the interval $[\ul{b}, a(\ol{X})]$. 
The latter strategy is precisely a generalised trailing stop, where the stochastic floor  increases along with the running maximum $\ol{X}$ (see also \cite{Tim_HZ18}). 

\item[(iii)] If the trailing stop order is not activated as $\overline{X}$ continuously increases towards the take-profit target log price $b^\star(y_c)$,  it will be optimal to sell the asset once  $X$ reaches $b^\star(y_c)$. 

\item [(iv)] As an independent case, not communicating with the aforementioned ones (i)--(iii), when the starting maximum log price $s$ is higher than $b^\star(y_c)$, we obtain similar economic insights as in the severe anxiety with high drawdown tolerance, since the optimal selling strategy is realised either at a traditional take-profit sale, before a new maximum is established, or at a stop-loss type order. 
\end{enumerate}

\begin{rmk}
It is worth  mentioning that, when setting a trailing stop type order, there is a possibility that the asset log price jumps downwards from the interval $(a(s),s]$ to the interval $(-\infty,\underline{b})$. In this case, we proved that it is optimal for the investor not to sell the asset, but rather be patient by waiting until $T^+_{\underline{b}}$ to sell (this is also true for all aforementioned cases under severe anxiety that require the use of a stop-loss type strategy). 
This reflects the additional protection sought by investors in financial markets through a trailing stop (resp., stop-loss) with a limit. The purpose is to secure a price only when the asset price experiences a drawdown from its peak that crosses the trailing stop (resp., stop-loss) threshold but not the limit. 
We can further interpret this result as missing out on selling the asset after a relatively big price jump, since the asset has already lost enough value that it ``costs" nothing to wait for some more time until the log price recovers to $\underline{b}$.
\end{rmk}

In order to illustrate the optimal selling strategy proposed in Theorem \ref{thm3}, we present two numerical case studies for identical assets: 
In Figure \ref{fig:trail}, the asset is optimally sold at time $t=6.7795$ at the trailing stop threshold $a(\bar{X}_t)=2.949$, protecting some of the ``profits'' from holding this asset, whose value rose significantly from the initial log price $x=2$;
in Figure \ref{fig:notrail}, the trailing stop is never activated (even though it is set up at some point), so the investor sells the asset at the take-profit threshold $b^\star(y_c)=4.1748$.

\begin{rmk}
We comment here that the rationale behind the consideration of the ODE \eqref{eq:ODE} is the imposition of Neumann condition on the value of the optimal selling strategy at the diagonal $\p\mc{O}_+$ of the two-dimensional state space.
Different from the majority of existing literature on optimal stopping problems involving the maximum process,  we do have a boundary condition at $s=b^\star(y_c)$ for the ODE, which helps us to obtain a unique solution $a(\cdot)$ as the candidate down-crossing sell order. 
When such a boundary condition is not available, an appropriate (unique) candidate must be chosen from the set of infinitely many solutions of the ODE, by relying on various different methods; 
e.g. using the transversality condition (see \cite{Guo2010}, \cite{RZ17watermark}, among others), or the maximality principle from \cite{peskir1998maximality} (see 
\cite{obloj2007maximality} for diffusion models, \cite{Kyp_Ott14} for L\'evy models, among others). 
\end{rmk}

\vspace{2pt}
The remaining Sections \ref{sec:comp}--\ref{sec:res} are devoted to proving the main results of the paper that have been presented in this section.

\subsection{Comparative statics}

Given that the structure of the stopping region is explicitly determined by model parameters via equations of differentiable functions, we know that the optimal stopping boundaries are continuous in $c$ and $q$.
It is easy to see that the set of optimal selling regions shrinks with the investor's tolerance level $c$ for asset price drawdowns and increases with the anxiety rate $q$ about these drawdowns. 
Namely, when %the investors' tolerance for asset price drawdowns 
$c$ decreases or when %their anxiety about drawdowns 
$q$ increases, investors should become more proactive and (optimally) sell their asset at lower profit-taking and/or higher stop-loss/trailing stop targets.    
These results can also be proved independently of the theory developed in Section \ref{mainres}, directly via the expression of the value function in \eqref{eq:problem1}. 
In particular, the result for $c$ is given in Proposition \ref{prop:compare1} in Section \ref{sec:LB} below, while one can similarly prove the monotonicity with respect to $q$ from (\ref{A}) and (\ref{eq:omega}), by observing that $q \mapsto R_t^{c}(q)$ is non-decreasing.\footnote{Here, we used the notation $R_t^{c} = R_t^{c}(q)$ to stress the dependence of the Omega clock on the parameter $q$.} 

We can further prove that the set of optimal selling strategies for investors with mild anxiety (cf.\ Theorem \ref{thm1} in Section \ref{sec:mild}) is strictly increasing with the investors' risk aversion coefficient $\rho$. 
In particular, we prove in Lemma \ref{lem:comp} that the mappings $\rho \mapsto \ul{b}(\rho)$, $\rho \mapsto z^\star(y;\rho)$ and $\rho \mapsto z_c(\rho)$ are all strictly decreasing, whenever they are defined (cf.\ Section \ref{mainres}).\footnote{\label{foot}Here, we used the notation 
$\ul{b} = \ul{b}(\rho)$,
$z^\star(y) = z^\star(y;\rho)$ and 
$z_c = z_c(\rho)$ 
to stress the dependence of these take-profit selling targets on the parameter $\rho$.}  
Therefore, %in the case of mild anxiety, 
more risk averse investors should be more proactive and (optimally) sell their asset at lower profit-taking targets.

Our analytical results in Section \ref{mainres} allow also for a numerical study of comparative statics with respect to general model parameter configurations, including cases of severe anxiety. 
To set up a numerical study for comparative statics of the risk-aversion coefficient $\rho$, volatility $\sigma$ and jump distribution intensity $\eta$, we consider an asset with log price process given by a compound Poisson jump process plus a Brownian motion with drift.
Namely, we consider an asset price model with Laplace exponent
\be \label{psib}
\psi(\beta)=0.18\beta+\frac{1}{2}\sigma^2\beta^2  - 0.25 \, \frac{\beta}{\beta+\eta}
\ee
within seven model parameter configurations, as shown in Table \ref{table:0} below.
  
\begin{table}[ht] 
\centering
\begin{tabular}{c|c c c c c  c} 
 \hline
Configuration & 
$\sigma$ & $\eta$ &  $r$ & $q$ & $\rho$ & $c$\\ [0.5ex] 
 \hline
 \text{Benchmark}  & 0.2 & 4 &  0.18 & 1 & 0.25 & 0.3568\\ 
 \text{Smaller $\rho$} & 0.2 & 4 &  0.18 & 1 & $\bm{0}$ & 0.3568\\
 \text{Larger $\rho$} & 0.2 & 4 &  0.18 & 1 & $\bm{0.5}$ & 0.3568\\
\text{Smaller $\sigma$} & $\bm{0.12}$ & 4 &  0.18 & 1 & 0.25 & 0.3568\\
\text{Larger $\sigma$} & $\bm{0.3}$ & 4 &  0.18 & 1 & 0.25 & 0.3568\\ 
\text{Smaller $\eta$}  & 0.2 & $\bm{2}$ &  0.18 & 1 & 0.25 & 0.3568\\
\text{Larger $\eta$} & 0.2 & $\bm{8}$ &  0.18 & 1 & 0.25 & 0.3568\\
[1ex] 
 \hline
\end{tabular}
\caption{Configurations of model parameters for comparative statics}
\label{table:0}
\end{table}
\begin{table}[ht]
\centering
\begin{tabular}{c|c  c c c} 
 \hline
Configuration & $\ul{b}$ & $b^\star(y_c)$ & $s_c$ & $b^\star(\hat{y})$ \\ [0.5ex] 
 \hline
 \text{Benchmark} & 0.2277 &  1.2793 & 0.7331 & 1.5593\\ 
 \text{Smaller $\rho$} & 0.2349 &  1.8713 & 0.7851 & 2.1489\\
 \text{Larger $\rho$} & 0.2211 &  1.0160 & 0.6966 &  1.3012\\
\text{Smaller $\sigma$} &  0.1792 & 1.1204 & 0.6950 & 1.4358\\
\text{Larger $\sigma$} &0.2987 & 1.5807 & 0.7964 & 1.8416\\ 
\text{Smaller $\eta$} & 0.2221 & 0.8555 & 0.6788 & 1.2123\\
\text{Larger $\eta$} & 0.2341 & 1.6568 & 0.7838 & 1.9333\\
[1ex] 
 \hline
\end{tabular}
\caption{Key selling thresholds for the problem configurations in Table \ref{table:0}}
\label{table:1}
\end{table}

We illustrate in Figure \ref{fig:comp} the optimal stopping regions for investors with severe anxiety and low drawdown tolerance (cf.\ Theorem \ref{thm3}), when perturbing their risk aversion coefficient $\rho$, the log price's volatility $\sigma$, or the intensity $\eta$ of the jumps' exponential distribution according to the configurations in Table \ref{table:0}. 
The key selling thresholds for each case are also shown in Table \ref{table:1} complementing the illustrations in Figure  \ref{fig:comp}.

\begin{figure}[h]
\centering
  \subfloat[Smaller $\rho$]{\includegraphics[width=0.3\textwidth]{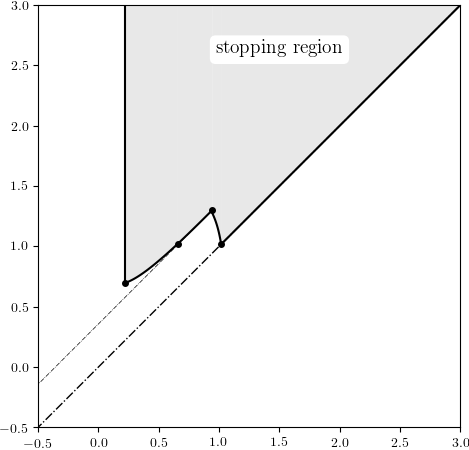} }
  \subfloat[Benchmark]{\includegraphics[width=0.3\textwidth]{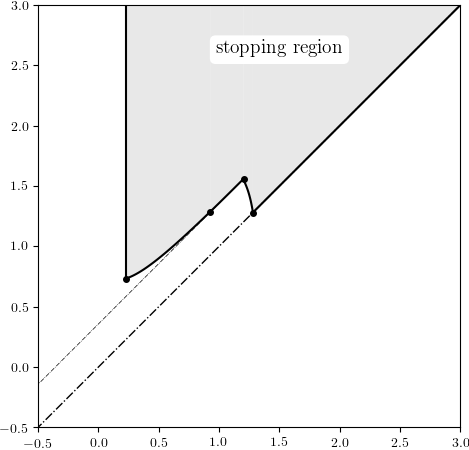} }
  \subfloat[Larger $\rho$]{\includegraphics[width=0.3\textwidth]{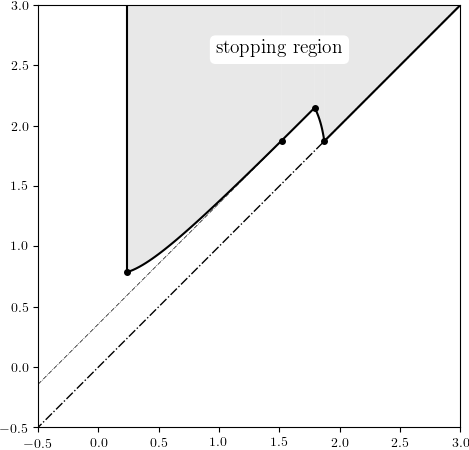} }
\\
  \subfloat[Smaller $\sigma$]{\includegraphics[width=0.3\textwidth]{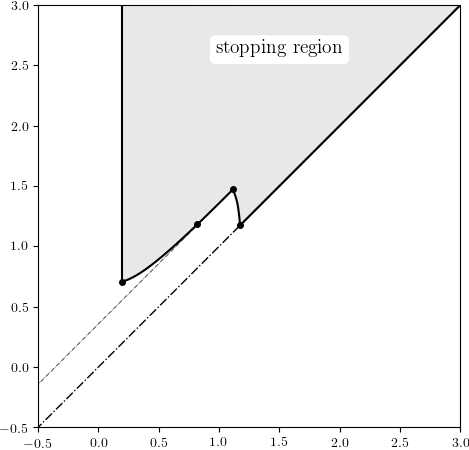} }
  \subfloat[Benchmark]{\includegraphics[width=0.3\textwidth]{benchmark.png} }
  \subfloat[Larger $\sigma$]{\includegraphics[width=0.3\textwidth]{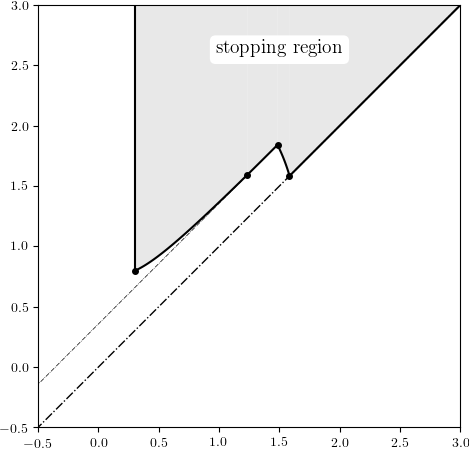} }
\\
  \subfloat[Smaller $\eta$]{\includegraphics[width=0.3\textwidth]{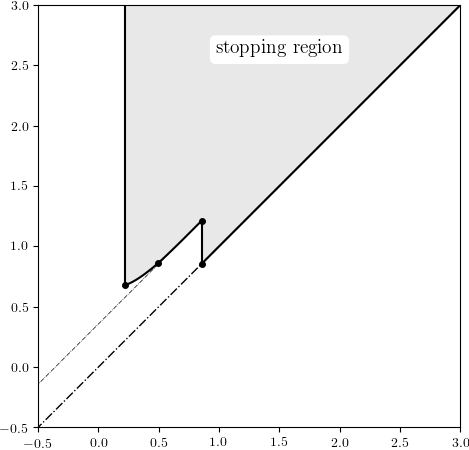} }
  \subfloat[Benchmark]{\includegraphics[width=0.3\textwidth]{benchmark.png} }
  \subfloat[Larger $\eta$]{\includegraphics[width=0.3\textwidth]{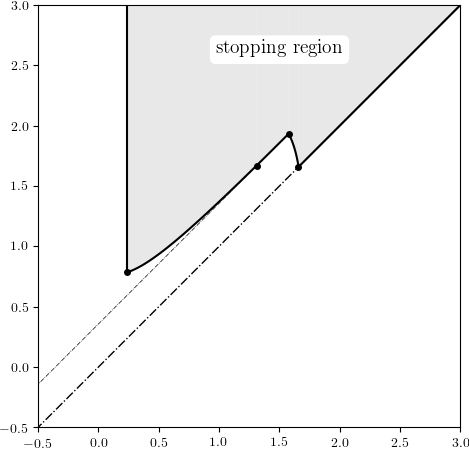} }
\caption{Stopping regions for the problem configurations in Table \ref{table:0} (see also Table \ref{table:1} for the values of key thresholds).}
\label{fig:comp}
\end{figure}

One can observe from Figure \ref{fig:comp} and Table \ref{table:1} that: 
$(a)$ the set of selling strategies appears to grow (resp., shrink) with respect to $\rho$ (resp., $\sigma$ and $\eta$) also in the case of severe anxiety;
$(b)$ the changes in the trailing stop $a(s)$, stop-loss $a^\star(s)$ and take-profit $\ul{b}$ targets, whenever they exist for each fixed $s>0$, are relatively minor under all parameters examined, compared to the resulting changes in the take-profit target $b^\star(y_c)$, which are much more significant.      
In summary, 
\begin{enumerate} 
\item for both degrees of anxiety, the more risk averse investors tend to be more proactive and (optimally) sell their asset at lower profit-taking and/or higher stop-loss (or trailing stop) targets. 
\item in cases of severe anxiety, the more volatile the asset price, the longer investors  wait before (optimally) selling their asset, as it is widely understood that options tend to gain value when volatility increases; 
\item in cases of severe anxiety, the higher the intensity of the exponential distribution, the smaller the size of negative jumps tends to be, hence the investors wait longer before (optimally) selling their asset. 
\end{enumerate}

\section{Bounds for the value function $V(x,s;c)$}
\label{sec:comp}

Fixing any tolerance level $c>0$, we denote by $\mc{S}_c$  the ``stopping region''  of log prices. This contains all the states of price and maximum price at which the investor should  (optimally) sell the asset. By the general theory of optimal stopping for Markov processes (see, e.g. \cite[Ch. I, Sec. 2.2]{PS}), we define 
\ben
\mc{S}_c:=\{(x,s)\in\mc{O}_+ \,:\, V(x,s;c)- U(\et^x)=0\}. 
\een 

As the first step of determining the value function $V(\cdot,\cdot;c)$ of (\ref{eq:problem1}), we derive a lower and an upper bound. We also provide bounds for the stopping region $\mc{S}_c$, which will be useful in proofs in later sections.

\subsection{Lower bound for $V(x,s;c)$ and upper bound for $\mc{S}_c$} \label{sec:LB}

For any $c'> c>0$,  it is easily seen from (\ref{A}) and (\ref{eq:omega}) that $R_t^{c'}\le R_t^{c}\le R_t^0  \leq (r+q)t$ holds for all $t\ge0$. 
Therefore, the value of \eqref{eq:problem1} is always bounded from below by  $\ul{v}(x)$, the optimal expected value of utility $U(\et^X)$ discounted at rate $r+q>0$. 

\begin{prop}\label{prop:compare1}
For any $c'> c>0$, 
\ben
\underline{v}(x) \leq V(x,s;c)\le V(x,s;c'), 
\quad \forall\,(x,s)\in\mc{O}_+,
\een
and 
\ben
\mc{S}_{c} \subseteq [\underline{b},\infty)^2\,\cap\,\mc{O}_+.
\een
\end{prop}
\begin{proof} For any fixed $(x,s)\in\mc{O}_+$, $ c'>c>0$ and any stopping time $\tau\in\timset$,  we have on the event $\{\tau<\infty\}$ that, the following inequalities 
\begin{align*}
\et^{-(r+q)\tau} U(\et^{X_\tau}) 
\le \et^{-R_\tau^c} U(\et^{X_\tau}) 
\le \et^{-R_\tau^{c'}} U(\et^{X_\tau}) 
\end{align*}
hold true ($\pr_{x,s}$-a.s.). 
Taking expectations $\ex_{x,s}$\footnote{Note that the $s$-component is actually redundant when the objective function does not involve the running maximum.} and the suprema over all stopping times $\tau$, we obtain in view of the definition \eqref{eq:underlinev} of $\underline{v}$ that 
\begin{align*}
\underline{v}(x)\leq  V(x,s;c)\le V(x,s;c'). 
\end{align*}
The bound for $\mc{S}_c$ follows immediately. 
\end{proof}

As a result of Proposition \ref{prop:compare1},  we may equivalently express our objective \eqref{eq:problem1} as
\be
\label{eq:problem2}
V(x,s;c) = \sup_{\tau\in \timset}\ex_{x,s}[\et^{-R_\tau^c} \, \ul{v}(X_{\tau}) \ind_{\{\tau<\infty\}}]
,\quad\forall\,(x,s)\in\mc{O}_+.
\ee
This will be a more convenient form of our problem in some parts of the forthcoming analysis (cf.\ Section \ref{highc}).

\subsection{Upper bound for $V(x,s;c)$ and lower bound for $\mc{S}_c$}

Consider the value function $\ol{v}(x;y)$ of \eqref{eq:oldproblem}, i.e. optimal expected value of utility $U(\et^X)$ with discounting $A^y$. 
We define the optimal stopping region of problem \eqref{eq:oldproblem}, (see, e.g. \cite[Ch. I, Sec. 2.2]{PS}) by 
\be
\mc{D}_y:=\{x\in\R: \ol{v}(x;y)-U(\et^x)=0\},\quad\forall\, y\in\R \,.
\label{Dy}
\ee
Problem \eqref{eq:oldproblem} can be considered as a simplified version of the original problem \eqref{eq:problem1}, 
when the discount factor does not update with the running maximum $\ol{X}$.
Specifically, for any fixed $c>0$ and $(x,s)\in\mc{O}_+$, the continuous additive functional $A^{s-c}$ is almost surely dominated from above by $R^c$ under $\pr_{x,s}$. 
Thus, the value function of \eqref{eq:problem1} is always bounded from above by the value of \eqref{eq:oldproblem} when $y=s-c$.

\begin{prop}\label{prop:compare}
For any fixed $c\ge0$, we have 
\be
V(x,s;c)\le \ol{v}(x;s-c),\quad \forall\,(x,s)\in\mc{O}_+.\label{eq:bounds_fixed}
\ee 
Then, the optimal stopping region $\mc{D}_y$ defined by \eqref{Dy} satisfies
\be
(\mc{D}_{s-c}\times\{s\})\,\cap\,\mc{O}_+\subseteq\mc{S}_c.\label{eq:bounds_fixed1}
\ee
Moreover, the equalities in \eqref{eq:bounds_fixed} and \eqref{eq:bounds_fixed1} hold if and only if $[s,\infty) \subseteq \mc{D}_{s-c}$.
\end{prop}
\begin{proof}
We only need to prove that the equality in \eqref{eq:bounds_fixed} holds if and only if  $[s,\infty) \,\subseteq \mc{D}_{s-c}$. 
But the latter condition means that it is optimal to stop in problem \eqref{eq:oldproblem} before the asset log price $X$ reaches $s$. 
This is equivalent to the running maximum $\ol{X}$ remaining constant, equal to $s$, in which case the value of the same strategy for problem \eqref{eq:problem1} under $\pr_{x,s}$ is the same as $\ol{v}(x;s-c)$. Therefore, the optimal value $V(x,s;c)$ for problem \eqref{eq:problem1} is no less than $\ol{v}(x; s-c)$. This completes the proof. 
\end{proof}

In the next section, we focus on solving for the upper bound $\ol{v}(x;y)$.

\section{Cracking problem \eqref{eq:oldproblem} with value function $\ol{v}(x;y)$} \label{sec:barv}

This section is concerned with the study of  problem \eqref{eq:oldproblem}, which servers as a cornerstone in the analysis of problem \eqref{eq:problem1}. 
Recall that, the case of risk neutral utility has already been treated in \cite[Theorem 2.4, Theorem 2.5]{OmegaRZ}. 

\begin{prop}\label{prop41}
The value function of \eqref{eq:oldproblem} satisfies the following properties:
\begin{enumerate}
\item[(i)] $\ol{v}(x;y)$ is strictly increasing and continuous in $x$ over $\R$, and is non-increasing and continuous in $y$ over $\R$;
\item[(ii)] if there exists a constant $a\in\mc{D}_y\cap(-\infty, y]$, then $y\ge\ul{b}$ and $[\ul{b},a]\subset\mc{D}_y$;
\item[(iii)] the optimal stopping region $\mc{D}_y$ is a union of disjoint closed intervals and there is at most one component  that lies in $(y,\infty)$.
\end{enumerate}
\end{prop}
\begin{proof}
For any $\rho\in[0,1)$, we know from
\[\frac{u^{1-\rho}-1}{1-\rho}\bigg|_{u=1}=0, \quad \frac{\partial}{\partial u}\bigg|_{u=1}\bigg(\frac{u^{1-\rho}-1}{1-\rho}\bigg)=1, \quad\text{and} \quad \frac{\partial}{\partial\rho}\frac{\partial}{\partial u}\bigg(\frac{u^{1-\rho}-1}{1-\rho}\bigg)=-u^{-\rho}\,\log u<0\quad\forall u>1,\]
that 
\[\bigg(\frac{u^{1-\rho}-1}{1-\rho}\bigg)^+\le (u-1)^+,\quad\forall \;u\in\R_+.\]
Hence, by the dominated convergence theorem, one can repeat the steps used in the proof of \cite[Proposition 3.1]{OmegaRZ} to prove $(i)$ and $(ii)$. 
The claim in the first half of $(iii)$ follows from the fact that $\ol{v}(x;y)$ is continuous in $x$ over $\R$; the second half of the claim in $(iii)$ can be proved in the same way as in \cite[Proposition A.1]{OmegaRZ}.
\end{proof}

We recall that, if $\ol{u}=-\infty$ or $\ol{u}\ge0$ and $y<\bar{y}$, the candidate up-crossing selling threshold $z^\star(y)$ of \eqref{z*}, is actually the largest root to \eqref{eq:zstarg}. By the monotone property of $g$ we know that $z^\star(y)-y$ is strictly decreasing. Moreover, Proposition \ref{prop41}(i) also implies that if $\mc{D}_{y}=[z^\star(y),\infty)$ for all $y$ in some interval $I$ with nonempty interior, then 
$z^\star(y)$  is continuous and non-increasing in $y$ over $I$. Proposition \ref{prop41}(ii)--(iii) imply that there are three possibilities for the stopping region:
$\rm(I)$ $\mc{D}_y$ does not include $\ul{b}$, 
or $\rm(II)$ $\mc{D}_y=[\ul{b},\infty)$, 
or $\rm(III)$ $\mc{D}_y=[\ul{b}, a]\cap[b,\infty)$ for some $\ul{b}\le a<y<b$.

In what follows, we study the problem \eqref{eq:oldproblem} separately in the two cases of mild %anxiety 
(i.e. $\bar u=-\infty$) and severe (i.e. $\bar u \ge0$) anxiety (see Section \ref{sec:connection}). The following results generalise \cite[Theorem 2.4, Theorem 2.5]{OmegaRZ} to the case of risk averse investors.

\subsection{Investors with mild anxiety}
\label{sec:barvm}

\begin{thm}
\label{thm00}
For an investor with mild anxiety (i.e. $\ol{u}=-\infty$), the optimal stopping region and the value function for problem \eqref{eq:oldproblem} are given by 
$\mc{D}_y=[z^\star(y),\infty)$, 
and
\be \label{v1}
\ol{v}(x;y)= \ind_{\{x<z^\star(y)\}} U(\et^{z^\star(y)}) \frac{\mc{I}^{(r,q)}(x-y)}{\mc{I}^{(r,q)}(z^\star(y)-y)}+\ind_{\{x\ge z^\star(y)\}} U(\et^x) ;
\ee
Moreover, the function $z^\star(y)$ of \eqref{z*} is continuous and strictly decreasing over $(-\infty,\ul{b}]$, and $z^\star(y)\equiv \ul{b}$ for all $y\ge\ul{b}$.
\end{thm}
\begin{proof} 
The expression for $\mc{D}_y$ has already been proved in the preliminary analysis of Section \ref{sec:connection}, which also implies the form of $\ol{v}(x;y)$ in view of Lemma \ref{lem:I}. 
The expression of $z^\star(y)$ for $y>\ul{b}$ follows from equations \eqref{eq:zstarg} and \eqref{b_}, using the fact that $\Lambda(x)\equiv\Phi(r+q)$ for all $x<0$. 
Also, Proposition \ref{prop41} and succeeding discussions imply the continuity and non-increasing property of $z^\star(y)$ over $\R$. 

To prove the strictly decreasing property of $z^\star(y)$ over $(-\infty, \ul{b}]$, suppose that there exist $y_1<y_2< \ul{b}$ such that $z^\star(y_1)=z^\star(y_2)\ge\ul{b}$, aiming for a contradiction. Then, we must have $\ol{v}(x;y_1)\equiv \ol{v}(x;y_2)$ for all $x\in\R$. However, it is easily seen that 
$$A_{T_{z^\star(y_1)}^+}^{y_1}<A_{T_{z^\star(y_2)}^+}^{y_2}, \quad \pr_x-a.s. \text{ for any $x> y_2$} 
\quad \Rightarrow \quad 
\ol{v}(x;y_1)>\ol{v}(x;y_2) \text{ for any $x>y_2$,} 
$$ 
which is a contradiction. 
Hence, $z^\star(y)$ must be strictly decreasing over $(-\infty, \ul{b}]$.
\end{proof}

\subsection{Investors with severe anxiety}
\label{sec:barvs}

\begin{thm}
\label{thm0}
For an investor with severe anxiety (i.e. $\ol{u}\ge0$), we have $\ul{b}<\tilde{y}<\ol{y}$ and $\tilde{y}<\hat{y}$. The optimal stopping region and the value function for problem \eqref{eq:oldproblem}
 are given as follows: 
\begin{enumerate}
\item[\rm(a)] if $y<\tilde{y}$, then $\mc{D}_y=[z^\star(y),\infty)$, 
while if $y=\tilde{y}$, then $\mc{D}_y=\{\underline{b}\}\,\cup\, [z^\star(y),\infty)$. 
The value function $\ol{v}(x;y)$ is given by \eqref{v1};
\item[\rm(b)] if $\tilde{y}<y<\hat{y}$, then $\mc{D}_y=[\underline{b},a^\star(y)] \,\cup\, [b^\star(y),\infty)$, where $a^\star(y)$ and $b^\star(y)$ are defined in \eqref{a*b*}, and 
\begin{align} \label{vflast}
\ol{v}(x;y) = \underline{v}(x) + \ind_{\{a^\star(y)<x<b^\star(y)\}} \Delta(x,a^\star(y);y)\,,
\end{align}
with  $\Delta(x, a; y)$ given in \eqref{D}. 
\item[\rm(c)] if $y\ge \hat{y}$, then $\mc{D}_y=[\underline{b},\infty)$ and $\ol{v}(x;y)=\underline{v}(x)$.
\end{enumerate}
Overall, the function $a^\star(\cdot)$ is continuous and strictly increasing over $[\tilde{y}, \hat{y})$, with $a^\star(\tilde{y})=\ul{b}$ and $\lim_{y\uparrow \hat{y}}a^\star(y)=\hat{y}$, while the mapping
\[y\mapsto\begin{dcases}
z^\star(y),\quad & y\le \tilde{y},\\
b^\star(y),& \tilde{y}<y<\hat{y},
\end{dcases}\]
is continuous and strictly decreasing, with $\lim_{y\uparrow \hat{y}} b^\star(y)=\hat{y}$;  (see Figure \ref{fig:example}(a) for a plot of $a^\star(y)$ and $b^\star(y)$).
\end{thm}

\begin{proof}
\ul{\it Construction of the critical  $y$-value $\bar{y}$ from \eqref{ybar}.}
In view of equation \eqref{eq:g}, the fact that $\ol{u}$ is a local minimum of $g(\cdot)$ and the monotonicity 
\ben 
\frac{\partial}{\partial y} \Big( \et^{(1-\rho)y} g(z-y) \Big) 
= 
- \et^{(1-\rho)z} \frac{\Lambda'(z-y)}{\Lambda(z-y)} > 0 \,
\een
which follows from $\Lambda(\cdot)$ being strictly decreasing on $\R_+$ (see \cite[Lemma 4.2]{OmegaRZ}), we conclude that there exists $\bar{y}$ such that $z^\star(\bar y) = \bar u +\bar y$, which is given by the expression \eqref{ybar}.
For any $y>\bar y$, the value $z^\star(y)=\bar u + y$ does not identify with a root of \eqref{eq:zstarg} and a candidate up-crossing threshold for problem \eqref{eq:oldproblem}.
We thus focus on the values of $y \le \bar y$.

\vspace{2pt}

\ul{\it Construction of the critical  $y$-value $\tilde{y}$ from \eqref{tily}.}
In view of the above, for any fixed $y\le \bar{y}$, the up-crossing threshold  $z^\star(y)$ of \eqref{z*} satisfies  $z^\star(y)\ge y+\ol{u}$. 
By employing the techniques leading to \cite[Proposition 4.7]{OmegaRZ} in our setting, we know that the positive function $\ol{v}(x;y)$ given by \eqref{v1} satisfies: 
\begin{enumerate}
\item the process $\big(\exp(-A_t^y) \ol{v}(X_t;y)\big)_{t\ge0}$ is a super-martingale; \vspace{3pt}
\item the process $\big(\exp(-A_{t\wedge T_{z^\star(y)}^+}^y) \ol{v}(X_{t\wedge T_{z^\star(y)}^+};y)\big)_{t\ge 0}$ is a martingale;
\item $\ol{v}(x;y)= U(\et^x)$ for all $x\ge z^\star(y)$ and 
$\frac{\p}{\p x} \ol{v}(x;y)|_{x=z^\star(y)-}= U'(\et^{z^\star(y)})$.
\end{enumerate}
Combining these properties with the classical verification method (see, e.g. proofs of theorems in \cite[Section 6]{Alili2005}), we can establish the optimality of the selling strategy $T_{z^\star(y)}^+$ by finally proving that $\ol{v}(x;y) \geq U(\et^x)$ for all $x<z^\star(y)$, or equivalently  
\ben
D(x;y) \leq 1 , \quad\forall \; x < z^\star(y) , 
\quad \text{where } \; D(x;y):=\frac{U(\et^x)}{\ol{v}(x;y)}.
\een
Recall that $D(x;y)=1$ for all $x\ge z^\star(y)$.  Since $g(x-y)$ is strictly increasing for $y+\ol{u}<x<z^\star(y)$, we know that 
\be
\frac{\p}{\p x} D(x;y) = \frac{\Lambda(x-y)}{\ol{v}(x;y)}\Big( \frac{1}{1-\rho}- \et^{(1-\rho)y} g(x-y) \Big)>0\,,\quad \forall \; y+\ol{u}<x<z^\star(y) \,. \label{dxD}
\ee
Hence, $D(x;y)<1$ for all  $x\in [y+\bar u, z^\star(y))$. However, since \eqref{eq:g} can have more than one solution, \eqref{dxD} implies that $D(\cdot;y)$ may not be always increasing over $(-\infty,y+\bar u)$, in which case it admits at least one local maximum in this region.
Also, for every fixed $y<\bar{y}$ and $x<z^\star(y)$, given that $z^\star(\cdot)$ and $\Lambda(\cdot)$ are both decreasing over $(-\infty,\bar{y})$ and $(0,\infty)$ (see \cite[Lemma 4.2]{OmegaRZ}), we get
\begin{align}
\frac{\partial}{\partial y}\log D(x;y)
=& -\frac{\frac{\p}{\p y} \ol{v}(x;y,z^\star(y)) + \frac{\p}{\p z} \ol{v}(x;y,z))|_{z=z^\star(y)}\cdot\frac{\diff z^\star(y)}{\diff y}}{\ol{v}(x;y,z^\star(y))}\nn\\
=& \left.\begin{cases} 
\Lambda(z^\star(y)-y) \,, & \text{if } x< y\\
\Lambda(x-y)-\Lambda(z^\star(y)-y)\,, & \text{if } y<x<z^\star(y)
\end{cases} \right\} > 0
\,, \nn
\end{align}
thus the function $D(x;\cdot)$ is strictly increasing over $(-\infty,y)$. This yields that the largest local maximum of $D(\cdot;y)$ (which must be in $(-\infty, y+\ol{u})$) is also increasing in $y$ and we can define
$\tilde{y} := \inf \{ y\leq \bar{y} \,:\, \sup_{x<y+\ol{u}} D(x;y)=1\}$, 
which is equivalent to the definition \eqref{tily}, in view of the expression \eqref{v1} of $\ol{v}$. 

For $y=\bar{y}$ (and $z^\star(\bar{y})=\bar u + \bar y$), since $\bar u$ is a local minimum of $g(\cdot)$, we know that for all sufficiently small $\epsilon>0$, we have $g(\ol{u}-\epsilon)> e^{(\rho-1)\bar{y}} / (1-\rho)$, so $\frac{\p}{\p x}D(x;\bar{y})_{x=z^\star(\bar{y})-\epsilon}<0$. This implies that $D(x;\bar{y})>D(z^\star(\bar{y}); \bar{y})=1$ for all $x$ in a sufficiently small left neighborhood of $z^\star(\bar{y})$. Hence, we conclude that $\tilde{y}<\bar{y}$.

\vspace{2pt}

\ul{\it Proof of part $\rm (a)$.} 
By the construction of $\tilde{y}$, we know that for any fixed $y <\tilde{y}$, 
\ben
D(x;y)<1 , \quad\forall x<z^\star(y)\,, 
\een
which implies that the selling strategy $T_{z^*(y)}^+$ is optimal for all  $y <\tilde{y}$.

If $y=\tilde{y}$, we may also conclude that the candidate value function $\ol{v}(\cdot;\tilde{y})$ of \eqref{v1} is the true value function. 
Moreover, by the above properties of $\tilde{y}$, we know that there exists a point $x_0 <\ol{u}+\tilde{y}$ satisfying   
\ben
D(x_0;\tilde{y}) = \sup_{x<\ol{u}+\tilde{y}} D(x;\tilde{y}) = 1 \,. 
\een
Essentially, $x_0$ is a second point (other than $z^\star(\tilde{y})$) where the  value function $\ol{v}(\cdot;\tilde{y})$ smoothly touches the reward function $U(\et^x)$. 
Then for all sufficiently small $x$, $x_0$ will be the optimal up-crossing threshold, so $x_0$ is also a stationary point of $D(\cdot;\tilde{y})$ and solves $(1-\rho)^{-1}=\et^{(1-\rho) \tilde{y}}g(x_0-\tilde{y})$.
By Proposition \ref{prop41}(iii),  we may further conclude that $x_0<\tilde{y}$, thus  
\ben 
\frac{1}{1-\rho} =\et^{(1-\rho) \tilde{y}}g(x_0-\tilde{y})
=\et^{(1-\rho) x_0} \frac{\prq-1+\rho}{(1-\rho) \prq} 
\quad \Rightarrow \quad 
x_0=\underline{b}\,.
\een 
In view of the above, we conclude that there is only one such $x_0$ and $\mc{D}_{\tilde{y}}=\{\underline{b}\}\cup[z^\star(\tilde{y}),\infty)$ with $\ul{b} < \tilde{y}$.

\vspace{2pt}

\ul{\it Construction of the critical $y$-value $\hat{y}$ from \eqref{Lv_1} and proof of part $\rm(c)$.} 
In order to show that the value function is given by $\ol{v}(\cdot;y)=\underline{v}(\cdot)$ for all sufficiently large $y$, we adopt the method of proof through variational inequalities (see e.g. \cite{OksendalSulemBook} or \cite[Lemma 4.10]{OmegaRZ} for the same problem under risk neutral utility). 
Using the explicit expression of $\ul{v}(x)$ in \eqref{eq:underlinev}, one can easily see that $(\mc{L}-r)\ul{v}(x)\equiv \chi(x)$ from \eqref{eq:57}  for any $x>\ul{b}$, where 
$\mc{L}$ is the infinitesimal generator of $X$. Precisely,  
 for all functions ${F}(\cdot)\in C^2(\R)$, $\mc{L}F$ is given by
\[
\mc{L}{F}(x)=\frac{1}{2}\sigma^2{F}''(x)+\mu {F}'(x) + 
\int_{-\infty}^{0}({F}(x+z)-{F}(x)-\ind_{\{z>-1\}}z{F}'(x))\Pi(\diff z).
\]

Given that $(\mc{L}-r)\ul{v}(x)=q\ul{v}(x)>0$ for all $x<\ul{b}$, we know that $\hat{y}$ defined by \eqref{Lv_1} is the smallest $y$-value such that $\underline{v}(\cdot)$ is super-harmonic with respect to the discount rate $r+q\ind_{\{x<y\}}$. It follows that $\ol{v}(x;y)\equiv \ul{v}(x)$ and $\mc{D}_y\equiv[\ul{b},\infty)$ for all $y\ge \hat{y}$, which proves part $(c)$.
Since $\mc{D}^{\tilde{y}}\neq [\ul{b},\infty)$, we may conclude that $\tilde{y}<\hat{y}$. 

\vspace{2pt}

\ul{\it Proof of part $\rm(b)$.} The analysis in part (a), the fact that $\mc{D}_y$ is increasing in $y$ and the above observation (see also \eqref{Lv_1}) that $(y,\hat{y})\,\subset(\mc{S}^y)^c$, dictates the consideration of the following pasting points, for all $y\in\,(\tilde{y},\hat{y})$:
\begin{align}
a^\star(y):=\sup\{x\in[\underline{k}, y]: \ol{v}(x;y)=U(\et^x)\}
\quad \;\; \text{and} \;\; \quad
b^\star(y):=\inf\{x\in[\hat{y}, z^\star(\tilde{y})]: \ol{v}(x;y)=U(\et^x)\}.\label{eq:def_bs}
\end{align}
In fact, using Proposition \ref{prop41}(ii), we may conclude that $[\underline{b},a^\star(y)]\in\mc{D}_y$. 
Hence, for any $x\in\,(a^\star(y),b^\star(y))$, we have
\begin{align}
\ol{v}(x;y)
&=\ex_x\Big[ \exp\big(-A_{T_{a^\star(y)}^-\wedge T_{b^\star(y)}^+}^y\big) \underline{v}\big(X_{T_{a^\star(y)}^-\wedge T_{b^\star(y)}^+}\big) \Big]. \label{eqeqeqeqeq}
\end{align}
Using the results from Lemma \ref{lem:I} and \cite[Proposition 4.12]{OmegaRZ}, we can rewrite \eqref{eqeqeqeqeq}, using $\chi(\cdot)$ defined in \eqref{eq:57}, as 
\begin{align} \label{V2V22}
\ol{v}(x;y)
=&\underline{v}(x) 
+\int_{a^\star(y)}^{y}\bigg(\frac{W^{(r,q)}(x,a^\star(y))}{W^{(r,q)}(b^\star(y),a^\star(y))}W^{(r,q)}(b^\star(y),w)-W^{(r,q)}(x,w)\bigg) \left(\chi(w)-q\ind_{\{w<y\}}\underline{v}(w)\right)\diff w \nn\\
&+\int_{y}^{b^\star(y)} \bigg(\frac{W^{(r,q)}(x,a^\star(y))}{W^{(r,q)}(b^\star(y),a^\star(y))}W^{(r)}(b^\star(y)-w)-W^{(r)}(x-w)\bigg)  \left(\chi(w)-q\ind_{\{w<y\}}\underline{v}(w)\right)\diff w. 
\end{align}

In view of this explicit formula, one can easily see that the mapping $x\mapsto \ol{v}(x;y)$ is in $C^1(a^\star(y),b^\star(y))$.
By exploiting the optimality of $a^\star(y)$ and $b^\star(y)$ and similar arguments to the proof of \cite[Proposition 4.13]{OmegaRZ}, we can show that the smooth fit condition always holds for $b^\star(y)$, while it holds for $a^\star(y)$ only when $X$ has unbounded variation. 
Using the above smoothness at the log prices $a^\star(y)$ and $b^\star(y)$ together with \eqref{eqeqeqeqeq}--\eqref{V2V22}, we can show that the optimal pair $(a^\star(y)$, $b^\star(y))$ from \eqref{eq:def_bs} necessarily solves the following system:
\ben
\Delta(b,a;y)=\frac{\p}{\p b}\Delta(b,a;y)=0,\quad \text{for } \; (a,b)\in[\underline{k},y)\times[\hat{y},z^\star(y)], 
\een
where $\Delta(\cdot,\cdot;y)$ is given by \eqref{D}. The rest of the proof, including showing that $\ol{v}(x;y)$ takes the form \eqref{vflast} and $a^\star(y)$, $b^\star(y)$ from \eqref{eq:def_bs} are given by \eqref{a*b*}, is identical to the one for \cite[Theorem 2.5(iv)]{OmegaRZ}, hence is omitted for brevity.

Finally, the continuity, limits and weak (not strict) monotone properties of the relevant mappings follow from Proposition \ref{prop41} and the succeeding discussions. The strict monotonicity can be proved similarly to the proof of Theorem \ref{thm00}. 
\end{proof}
 
\begin{figure}%
\centering
\subfloat[$b^\star(y)$ and $a^\star(y)$]{{\includegraphics[width=200pt]{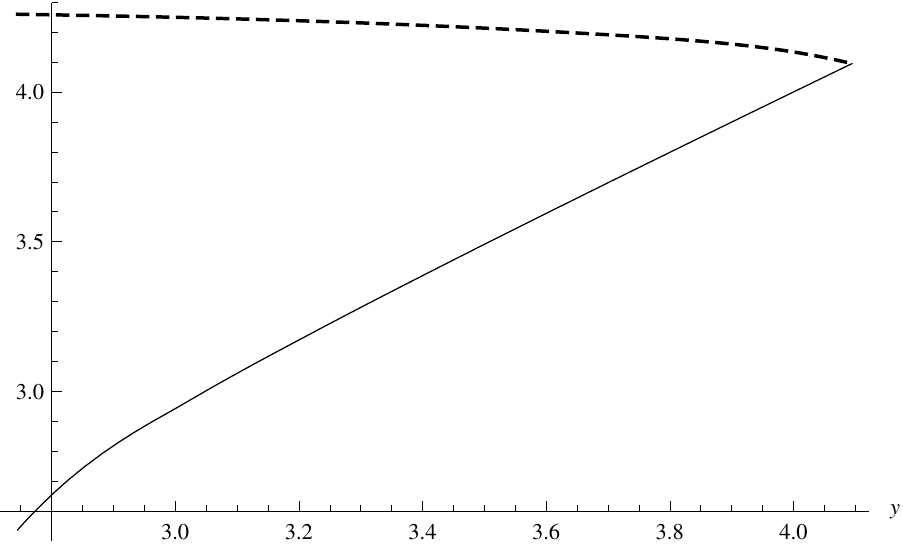} }}%
\qquad
\subfloat[$m(y)$]{{\includegraphics[width=200pt]{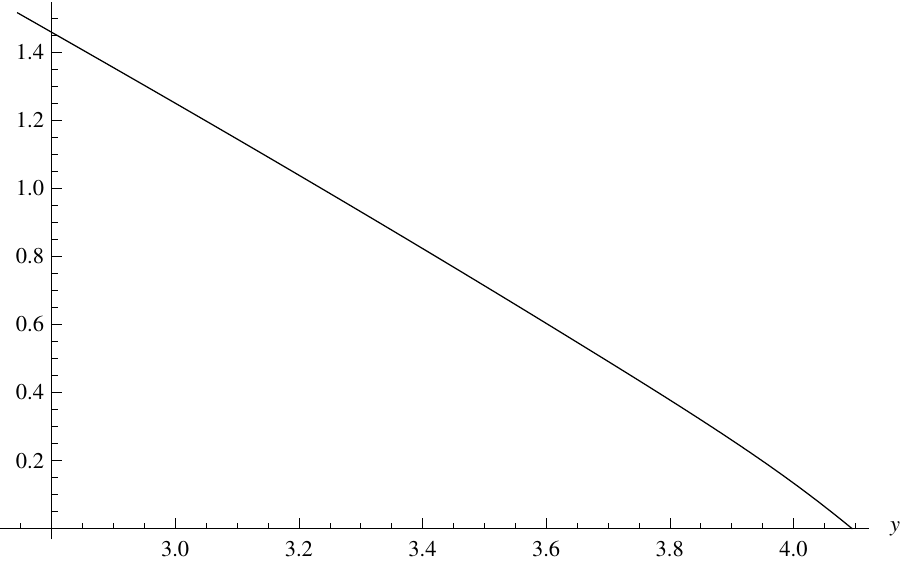} }}%
\caption{Using a compound Poisson plus a drifted Brownian motion model with $H^\star<0$ from \eqref{eq:explicit}, thus severe anxiety ($\ol{u}>0$), we plot the optimal boundaries $b^\star(y)$ (in dashed line) and $a^\star(y)$ defined in Theorem \ref{thm0} in figure (a); In figure (b), we plot the function $m(y)$ defined in \eqref{eq:m}. Both $b^\star(y)-a^\star(y)=0$ and $m(y)=0$ occur at $y=\hat{y}$.}%
\label{fig:example}%
\end{figure}

\section{Proofs of the main results}
\label{sec:res}

In this section, we shall prove that the solution to the problem \eqref{eq:problem1}, takes the forms presented in Theorems \ref{thm1}, \ref{thm2} and \ref{thm3}.
Our main verification approach is through the Hamilton-Jacobi-Bellman equation. 
Specifically, suppose that we can find a function  
$w:\mc{O}_+ \rightarrow \,(0,\infty)$ in $C^{1,1}(\mc{O}_+) \cap C^{2,1}(\mc{O}_+ \backslash\{(\theta_1,s),\ldots,(\theta_k,s)\})$ (resp., $C^{0,1}(\mc{O}_+)$) if $X$ has paths of unbounded (resp., bounded) variation, for some $\theta_1,\ldots,\theta_k \in\R$, 
such that $w(x,s) \geq U(\et^x)$ and is super-harmonic. That is, $w$ satisfies the variational inequality 
\be \label{verifeq}
\max\big\{ (\mc{L}-r-q\ind_{\{x<s-c\}})w(x,s) \,,\, U(\et^x) - w(x,s) \big\} = 0, 
\quad\forall \,(x,s)\in\mc{O}_+ \backslash\{(\theta_1,s),\ldots,(\theta_k,s),(s,s)\} \,,
\ee
with boundary condition 
\be \label{verifeq2}
\frac{\p}{\p s}\bigg|_{x=s} w(x,s) = 0 \,.
\ee
Using these properties of $w$ together with the It\^o-L\'evy lemma and the compensation formula, 
we know that $V(x,s;c)\equiv w(x,s)$ is the value function of the problem \eqref{eq:problem1} for this tolerance level $c>0$. 

\vspace{1pt}
Before we move on to the individual proofs, we define the mapping 
\be
m(y):=\begin{cases}
z^\star(y)-y, \quad \forall \,y\le \ul{b}, &\text{if }\ol{u}=-\infty;\\
z^\star(y)\ind_{\{y\le\tilde{y}\}}+b^\star(y)\ind_{\{\tilde{y}<y<\hat{y}\}}-y, \quad \forall \, y< \hat{y}, &\text{otherwise},
\end{cases}\label{eq:m}
\ee
where $z^\star(\cdot)$ and $b^\star(\cdot)$ are given by \eqref{z*} and \eqref{a*b*}. From  Theorems \ref{thm00} and \ref{thm0}, we know that function $m(y)$ is strictly decreasing over its domain, with a range equal to $\R_+$. A plot of $m(y)$ is shown in Figure \ref{fig:example}(b).

\subsection{Investors with mild anxiety}
\label{sec:51}

\begin{proof}[Proof of Lemma \ref{lem21}] 
The first claim is already proved in Theorem \ref{thm00}, while the second claim follows straightforwardly by definition \eqref{z_c} and \eqref{z*}--\eqref{eq:zstarg}. 
\end{proof}
\begin{proof}[Proof of Theorem \ref{thm1}] 
\ul{\it Proof of part for $s \geq z_c$.}
We know from Theorem \ref{thm00} under the case of $\ol{u}=-\infty$, that the optimal stopping region $\mc{D}_y$ for problem \eqref{eq:oldproblem} is the half-line 
$[z^\star(y),\infty)$. 
On the other hand, Proposition \ref{prop:compare} asserts that 
\ben 
\text{if} \quad [s,\infty) \,\subseteq \mc{D}_{s-c} \equiv [z^\star(s-c),\infty)\,, 
\quad \text{then} \quad V(x,s;c)=\ol{v}(x;s-c)\,. 
\een 
It is thus natural to consider the critical $s$-value, such that $\mc{D}_{s-c} = [s,\infty)$, namely, 
\be
 \quad z^\star(s-c)=s 
\quad \Leftrightarrow \quad m(s-c)=c\,, \label{y_c}
\ee
where $m(\cdot)$ is defined by \eqref{eq:m}. 
By the construction of $m(\cdot)$ we know that
\ben 
\exists \, ! \,\, y_c\in\,(-\infty,\underline{b} )
\quad \text{that solves} \quad 
m(y_c)=c\,.
\een
In fact, using \eqref{Lg} and \eqref{z*}, one can obtain an explicit expression for $y_c$:
\be\label{eq:zc2}
y_c 
=\frac{1}{1-\rho}\log\bigg(\frac{\Lambda(c)}{\Lambda(c)-1+\rho}\bigg)-c. 
\ee
Therefore, we can define $z_c$ as the $s$-value determined by \eqref{y_c}, which is given by \eqref{z_c}.
Taking into account the definition and expression of $z_c$ in \eqref{z_c}, we observe from the monotonicity of $\mc{D}_y$ in $y$ (cf. Theorem \ref{thm00}) that 
\ben
\text{if} \quad s\ge z_c \text{ (or equivalently, }s-c\ge y_c)\,, \quad \text{then} \quad 
[s,\infty)\,\subseteq[z_c,\infty)\,\equiv \mc{D}_{y_c}\subseteq\mc{D}_{s-c}\,.
\een 
Hence by Proposition \ref{prop:compare}, we know that 
$V(x;s;c)\equiv \ol{v}(x;s-c)$ for any $(x,s) \in \mc{O}_+$ such that $s\ge z_c$, where $\ol{v}$ admits the expression \eqref{v1}. 
The result follows by finally observing that, for all $x\in[z^\star(s-c),s]$, we have
\be 
\frac{\p}{\p s}\bigg|_{x=s} {V}(x,s;c)= \frac{\p}{\p s}\bigg|_{x=s} U(\et^x)= 0, \quad\forall s > z_c 
\quad \Rightarrow \quad 
{V}(\cdot,\cdot;s) \text{ satisfies \eqref{verifeq2} for all } s > z_c  \,.
\nn
\ee 

\vspace{1pt}

\ul{\it Proof of part for $s<z_c$.}
For the remaining case, we prove that it is optimal to wait until the process $(X,\overline{X})$ reaches the point $(z_c, z_c)$, i.e. to sell at $T_{z_c}^+$. 
The value of such strategy, denoted by $\bar V(x,s;c)$, is given by 
\begin{align} \label{U}
\bar V(x,s;c) 
&:= \ex_{x,s}\left[ \exp(-R_{T^+_{z_c}}^c) \, U\big( \exp(X_{T_{z_c}^+})\big) \ind_{\{T_{z_c}^+<\infty\}}\right] 
= \ex_{x,s}\left[\exp(-R_{T^+_s}^c) \bar V(s,s;c) \ind_{\{T^+_s<\infty\}}\right] \\
&= \ex_{x,s}\left[\exp(-A_{T^+_s}^{s-c}) \bar V(s,s;c) \ind_{\{T^+_s<\infty\}}\right]
\, \overset{\eqref{T+b}}{=} \,\dfrac{\mc{I}^{(r,q)}(x-s+c)}{\mc{I}^{(r,q)}(c)} \, \bar V(s,s;c) 
\,, \quad \text{for all } x<s< z_c \,, \nn
\end{align} 
thanks to the strong Markov property of $(X,\overline{X})$.
Using the Neumann boundary condition $\frac{\p}{\p s}|_{x=s} \bar V(x;s;c)=0$, 
as we expect the process $(X,\overline{X})$ to reflect at the diagonal of the state space $\p\mc{O}_+$ until it reaches the point $(z_c, z_c)$, 
we obtain from \eqref{U} that
\be \label{Us<zc}
\bar V(x,s;c) = 
\et^{-\Lambda(c)(z_c-s)} \, \dfrac{\mc{I}^{(r,q)}(x-s+c)}{\mc{I}^{(r,q)}(c)}\, U(\et^{z_c}) 
\,,\quad \forall\, (x,s)\in\mc{O}_+  \text{ s.t. } s< z_c.
\ee
Notice that, the positive function given by the right-hand side of \eqref{eq:thm1} is precisely $\bar{V}(x,s;c)$. 
Then by construction, it is easily seen that the mapping $(x,s)\mapsto \bar{V}(x,s;c)$ is continuous over $\mc{O}_+$, $x\mapsto \bar{V}(x,s;c)$ is $C^1$ over $(-\infty,s)$ for each fixed $s$, and it satisfies the Neumann condition \eqref{verifeq2} for all $s< z_c$. 
In addition, one can show that the Neumann condition also holds at $s=z_c$ by respectively computing the left and the right derivative with respect to $s$. 
Hence, it remains to show that $\bar{V}$ solves the variational inequality \eqref{verifeq} for all $(x,s)\in\mc{O}_+ \backslash\{(s,s)\}$.

We first prove that the inequalities involving the infinitesimal generator hold.
To this end, we fix $x < s\le z_c$ and by using the definition of $\bar V$ from \eqref{U} and similar arguments to Section 4 of \cite{Biffis_Kyprianou_2010}, we know that 
\be
(\mc{L}-(r+q\ind_{\{x<s-c\}})) \bar{V}(x,s;c)=0,\quad\forall \, x<s\le z_c.\label{eq:L1}
\ee

In order to finish the proof, we show in the sequel that the inequalities involving the dominance of $\bar{V}(x,s;c)$ over the intrinsic value $U(\et^x)$ hold as well. 
To that end, for $x\le s\le z_c$, we define
$$R(x,s;c):= \frac{U(\et^x)}{\bar{V}(x,s;c)} \,.$$ 
Taking the partial derivative of $R$ with respect to $x$ yields that 
\be
\frac{\p}{\p x}R(x,s;c)=\frac{\Lambda(x-s+c)}{\bar{V}(x,s;c)} \left(\frac{1}{1-\rho} - \et^{(1-\rho)(s-c)} g(x-s+c)\right),
\; \text{for $g$ and $\Lambda$ defined by (\ref{Lg})}. 
\label{eq:Rpartial}
\ee
Considering that $x \mapsto g(x-s+c)$ is strictly increasing over $(-\infty,s]$ when $\bar u = -\infty$, together with the fact that $z=z_c$ and $y=y_c$ from \eqref{z_c} and \eqref{eq:zc2}, respectively, solve the equation \eqref{eq:g}, we conclude that 
\begin{align*}
\frac{1}{1-\rho} - \et^{(1-\rho)(s-c)} g(x-s+c)
&\geq \frac{1}{1-\rho} - \et^{(1-\rho)(s-c)} g(c) 
= \frac{1}{1-\rho} - \et^{(1-\rho)s} \bigg(\frac{1}{1-\rho}-\frac{1}{\Lambda(c)}\bigg) \\
&\geq \frac{1}{1-\rho} - \et^{(1-\rho) z_c} \bigg(\frac{1}{1-\rho}-\frac{1}{\Lambda(c)}\bigg)
=0,\quad\forall \, x\le s<z_c.
\end{align*}
It follows that $x\mapsto R(x,s;c)$ is strictly increasing over $(-\infty,s]$, for any fixed $s< z_c$. Hence, $R(x,s;c)\le R(s,s;c)$ for all $x\le s<z_c$. Furthermore, we notice from the expression \eqref{Us<zc} of $\bar{V}$ that
\be \label{Rss}
R(s,s;c)=\frac{U(\et^s)}{\et^{-\Lambda(c)(z_c-s)} \, U(\et^{z_c})}.
\ee
Taking the derivative of the above with respect to $s$ and using \eqref{z_c} yields that
\begin{align*}
\frac{\diff}{\diff s} R(s,s;c)
&=\frac{\et^{\Lambda(c)(z_c-s)}\Lambda(c)}{U(\et^{z_c})}\bigg(\frac{1}{1-\rho} - \et^{(1-\rho) s} \bigg(\frac{1}{1-\rho}-\frac{1}{\Lambda(c)}\bigg)\bigg) \\
&>\frac{\et^{\Lambda(c)(z_c-s)}\Lambda(c)}{U(\et^{z_c})}\bigg(\frac{1}{1-\rho} - \et^{(1-\rho) z_c} \bigg(\frac{1}{1-\rho}-\frac{1}{\Lambda(c)}\bigg)\bigg)=0,
\quad\forall \, s<z_c.\nn
\end{align*}
In view of $R(z_c, z_c; c) = 1$, which follows straightforwardly from \eqref{Rss} for $s=z_c$, we thus obtain 
\be
\frac{U(\et^x)}{\bar{V}(x,s;c)}=R(x,s;c)\le R(s,s;c)<R(z_c,z_c;c)=1,\quad\forall \, x\le s<z_c,\nn
\ee
which implies that 
\be
\bar{V}(x,s;c)\ge U(\et^x) ,\quad\forall \, x\le s\le z_c.\label{eq:D2}
\ee

Combining \eqref{eq:L1} and \eqref{eq:D2}, we conclude that the variational inequality \eqref{verifeq} is satisfied and consequently we prove the optimality of $\bar{V}(x,s;c)$ for all $s\le z_c$.
\end{proof}

\subsection{Investors with severe anxiety}
\label{highc}

\begin{proof}[Proof of Lemma \ref{lem23}]
Part $(i)$, the detailed constructions of $\bar{y}$, $\tilde{y}$ and $\hat{y}$ and the inequalities in part $(ii)$, as well as part $(iv)$ are already proved in Theorem \ref{thm0}. Part $(iii)$ follows straightforwardly from definitions \eqref{z_c}, \eqref{tily} and \eqref{z*}--\eqref{eq:zstarg}. 
\end{proof}

\begin{proof}[Proof of Theorem \ref{thm2}]
\ul{\it Construction of $\tilde{c}$ from \eqref{tilc}.}
Recall from Theorem \ref{thm0}, under $\ol{u}\ge0$, that the optimal strategy for the simplified problem \eqref{eq:oldproblem} changes qualitatively when the value $y$ is less or greater than the key level $\tilde{y}$. 
Naturally, one may consider the tolerance level $\tilde{c}$ associated with $\tilde{y}$ through the definition
\be
\tilde{c}:=m(\tilde{y}), \quad \text{where $m(\cdot)$ is defined by (\ref{eq:m})}, 
\label{tilc0}
\ee 
which is equivalent to the definition in \eqref{tilc}.
By the construction of $m(\cdot)$ we know that $\tilde{c}$ is uniquely defined and positive. 

\vspace{1pt}

\ul{\it Proof of part for $s \geq z_c$.}
Following the same reasoning as in Section \ref{sec:51}, we consider the 
critical $s$-value satisfying the properties in \eqref{y_c} and thus solving 
$m(s-c)=c$. By the construction of $m(\cdot)$,
\be \label{yc2}
\exists \, ! \,\,  y_c \in\,(-\infty, \hat{y}] \text{ that solves } m(y)=c \,. 
\ee
However, in this case, we can further conclude from the monotonicity of 
$m(\cdot)$, the assumption $c\ge\tilde{c}$,  
the definition \eqref{tilc0} of $\tilde{c}$ and the fact that $\tilde{y}<\hat{y}$ that
\ben
m(y_c)=c\ge\tilde{c}=m(\tilde{y})
\quad \Leftrightarrow \quad y_c \leq \tilde{y}.
\een 
Thus, by Theorem \ref{thm0}.(a), we can define the critical $s$-value as $z_c :=y_c+c$ with $y_c$ given explicitly by \eqref{eq:zc2}, yielding that $z_c$ is indeed given by \eqref{z_c} in this case as well. 
Moreover, recall from Theorem \ref{thm0}.(a) that, the optimal stopping 
region $\mc{D}_{y_c}$ for problem \eqref{eq:oldproblem} is given by
\ben 
\mc{D}_{y_c} = \begin{cases} 
[z_c,\infty)\,, & \text{ if } y_c<\tilde{y} \quad \Leftrightarrow \quad c>\tilde{c}, \\
\{\underline{b}\}\U[z_c,\infty) \,, & \text{ if } y_c=\tilde{y} \quad \Leftrightarrow \quad c=\tilde{c} 
\quad (\text{where } z_c=z_{\tilde{c}}=z^\star(\tilde{y})).
\end{cases} 
\een
In both cases, using the monotonicity of $\mc{D}_y$ in $y$ (cf. Theorem \ref{thm0}), it is straightforward to see that 
\ben
\text{if} \quad s\ge z_c \,, \quad \text{then} \quad 
[s,\infty)\,\subseteq[z_c,\infty)\,\subseteq\mc{D}_{y_c}\subseteq\mc{D}_{s-c}\,.
\een
Hence, by Proposition \ref{prop:compare} we know that $V(x,s;c)\equiv \ol{v}(x;s-c)$ for any $(x,s)\in\mc{O}_+$ such that $s\ge z_c$, where $\ol{v}$ admits the expressions in Theorem \ref{thm0} for $y=s-c$. 
The result follows by the straightforward observation that the Neumann condition \eqref{verifeq2} holds for all $s \geq z_c$ (see e.g. Section \ref{sec:51} for similar arguments). 

\vspace{1pt}

\ul{\it Proof of part for $s<z_c$.}
For the remaining case, we prove that it is optimal to wait until the process $(X,\overline{X})$ reaches the point $(z_c,z_c)$. 
Notice that the positive function given by the right-hand side of \eqref{eq:thm2} identifies with $\bar{V}(x,s;c)$ from \eqref{U} (see also \eqref{Us<zc}), and  the proof follows similar arguments as the ones in the proof for Theorem \ref{thm1}.
Therefore, the only non-trivial task, before establishing the optimality of $\bar{V}$, is to prove the dominance of $\bar{V}$ over the intrinsic value $U(\et^x)$.
We examine below the ratio $R(x,s;c)=U(\et^x)/\bar{V}(x,s;c)$ for $s< z_c$. 

When $\bar u > 0$, we know that $g(\cdot)$ is non-monotone anymore, hence we cannot draw any conclusions from the partial derivative \eqref{eq:Rpartial} of $R(x,s;c)$. 
To this end, we employ a different technique than in the proof of Theorem \ref{thm1}.
We begin by noticing that for $s< z_c$, we have by \eqref{z_c} and \eqref{eq:zc2} that $s-c< z_c-c = y_c$, thus 
\be
m(s-c)\ge m(y_c)=c\ge\tilde{c}=m(\tilde{y}) \quad \Leftrightarrow \quad  
s-c\le y_c\le \tilde{y}. \label{eq:710}
\ee
Therefore, in light of Theorem \ref{thm0}.(a), we have 
\ben
z^\star(s-c)\ge z^\star(y_c) = z_c\ge z^\star(\tilde{y})>\underline{b}. 
\een
Then, we take the partial derivative of the expression  \eqref{Us<zc} of $\bar{V}(x,s;c)$ with respect to $s$, for $x \leq s < z_c$, 
\ben
\frac{\partial}{\p s} \bar{V}(x,s;c) 
= \left.\begin{cases} 
\bar{V}(x,s;c)  \left( \Lambda(c) - \Phi(r+q) \right) \,, & \text{for } x\in\,(-\infty, s-c), \\
\bar{V}(x,s;c)  \left( \Lambda(c) - \Lambda(x-s+c) \right) \,, & \text{for } x\in[s-c, s),
\end{cases} \right\} >0, \quad \forall \; x<s\,,
\een
due to the fact that $\Lambda(\cdot)$ is strictly decreasing on $\R_+$ (cf. \cite[Lemma 4.2]{OmegaRZ}). 
We therefore have 
\ben
\bar{V}(x,s;c) \geq \lim_{s \uparrow z_c} \bar{V}(x,s;c) 
= \dfrac{\mc{I}^{(r,q)}(x-z_c+c)}{\mc{I}^{(r,q)}(c)}\, U(\et^{z_c}) 
\overset{\eqref{z_c},\eqref{eq:zc2}}{=} \dfrac{\mc{I}^{(r,q)}(x-y_c)}{\mc{I}^{(r,q)}(z_c - y_c)}\, U(\et^{z_c}) .
\een
In light of \eqref{eq:710} and Theorem \ref{thm0}.(a) for $y=s-c \leq y_c \leq \tilde{y}$, we know from the above and \eqref{v1} that 
\ben
\bar{V}(x,s;c) 
\geq \dfrac{\mc{I}^{(r,q)}(x-y_c)}{\mc{I}^{(r,q)}(z_c - y_c)} \, U(\et^{z_c}) 
= \ol{v}(x;y_c) \geq U(\et^{x})
, \quad \forall \,(x,s)\in\mc{O}_+  \text{ s.t. } s< z_c \,.
\een

Finally, we comment that the only possibility for $\bar{V}(x,s;c)=U(\et^{x} )$ is realised when the two inequalities above are equalities. In particular, this can only occur when either $x=s \to z_c$, or when $s \to z_c=z^\star(\tilde{y})$ and $x=\ul{b}$ (see Theorem \ref{thm0}.(a) for $y=s-c = y_c = \tilde{y}$). 
\end{proof}

\begin{proof}[Proof of Lemma \ref{lem:as}]
The existence and the uniqueness of the solution follow from classical results for nonlinear ODEs. To show the other properties, let $\theta(s)=s-a(s)$, and let $F(s,a)$ be the slope field of \eqref{eq:ODE}, i.e. 
\[
F(s,a)=\dfrac{qW^{(r)}(c)}{W^{(r,q)}(s,a;s-c)} \, \dfrac{(1-\rho) \,\big(\underline{v}(s-c) + \Delta(s-c,a(s);s-c)\big)}{(r+q-\psi(1-\rho)) \et^{(1-\rho)\, a(s)} - (r+q) - (1-\rho) \, f(a(s))}
\]
Then $\theta(\cdot)$ satisfies the ordinary differential equation
\ben
\theta'(s)=1-F(s,s-\theta(s)). 
\een
By equation (\ref{newW}), we see that for $a=y=s-c$, which occurs whenever $\theta(s)=c$, for some $s\ge\underline{b}+c$, we get 
\be
F(s,s-c) =
q \, \dfrac{(1-\rho) \,U(\et^{s-c})}{(r+q-\psi(1-\rho)) \et^{(1-\rho)\,(s-c)} - (r+q) - (1-\rho) \, f(s-c)} \,,
\nn
\ee
since $\Delta(a,a;a)=0$ by its definition \eqref{D}.
Hence, for such values of $s$, we have 
\begin{align}
1-F(s,s-c) = 
\dfrac{(r-\psi(1-\rho)) \et^{(1-\rho)\,(s-c)} - r - (1-\rho) \, f(s-c)}{(r+q-\psi(1-\rho)) \et^{(1-\rho)\,(s-c)} - (r+q) - (1-\rho) \, f(s-c)}
\label{thetadecr}
\end{align}
However, from \eqref{z_cn} we have $b^\star(y_c) -y_c = c$, thus we get in this case that $\underline{b}\le s-c\le b^\star(y_c)-c = y_c<\hat{y}$.
Therefore, the inequality $\chi(x)>0$ for all $x<\hat{y}$ (see \eqref{Lv_1})
yields that  
\be
(r-\psi(1-\rho)) \et^{(1-\rho)\,(s-c)} - r - (1-\rho) \, f(s-c) < 0 \,.
\label{eq:719}
\ee
On the other hand, following the proof of \cite[Lemma 4.9]{OmegaRZ} one can show that $f(\cdot)$ is decreasing and continuous over $[\underline{b},\infty)$, with limit
\be
f(\ul{b})=\frac{\et^{(1-\rho)\ul{b}}}{1-\rho}\left(r+q-\psi(1-\rho)\right)-\frac{r+q}{1-\rho}-\half\sigma^2\Phi(r+q).\label{eq:fff}
\ee
Therefore,
\begin{align} \label{ineq1}
&(r+q-\psi(1-\rho)) \et^{(1-\rho)\,(s-c)} - (r+q) - (1-\rho) \, f(s-c) \\
&> (r+q-\psi(1-\rho)) \et^{(1-\rho)\,\ul{b}} - (r+q) - (1-\rho) \, f(\ul{b}) 
= \half(1-\rho)\sigma^2\Phi(r+q) \geq 0 \,, 
\quad \text{for all } \, s>\underline{b}+c , 
\nn
\end{align}
where the last equality follows from \eqref{eq:fff}. 
From \eqref{thetadecr}, \eqref{eq:719} and \eqref{ineq1} we know that 
\ben
\text{for all $s\leq z_c$ such that $\theta(s)=c$, we have $\theta'(s)=1-F(s,s-\theta(s))<0$.}
\een 
Thus, if there exists $s_0\in[ \underline{b}+c, b^\star(y_c))$ such that $\theta(s_0)=c$, then $\theta(s)<c$ for all $s\in(s_0,b^\star(y_c)]$. 
However, notice that 
\ben
\theta(z_c) = \theta(b^\star(y_c)) = b^\star(y_c)-a(b^\star(y_c))=b^\star(y_c)-a^\star(y_c)>b^\star(y_c)-y_c=c \,,
\een
which is a contradiction. 
Therefore, such an $s_0$ cannot exist and the only possibility is $\theta(s)>c$, i.e. $a(s)<s-c$, as long as $s\in[\underline{b}+c,b^\star(y_c))$ and $a(s)$ is well-defined. 

In the final part of the proof, we use the aforementioned property of $a(s)<s-c$, in order to examine the behaviour of $a(\cdot)$ when $a(s)>\underline{b}$, which also implies that $s-c>\underline{b}$.  
Then, for any $s$ fixed and all $x\in\,(a(s),s-c]$, 
\ben 
\ul{v}(x) + \Delta(x,a(s);s-c) = U(\et^{x}) + \Delta(x,a(s);s-c) > U(\et^{x}) > U(\et^{a(s)}) > U(\et^{\,\underline{b}}) > 0 .
\een
Combining all the above with the probabilistic meaning of $W^{(r,q)}(s,a(s);s-c)$ in Lemma \ref{lem:bigW}, we conclude that $a'(s)=F(s,a(s))>0$ for all $s>\ul{b}+c$, as long as $a(s)>\underline{b}$. 
We now define
\[
s_c:=\sup\{s<b^\star(y_c)\,:\, a(s)\le\underline{b}\}\,. 
\]
Notice that $a(b^\star(y_c))=a^\star(y_c)>a^\star(\tilde{y})=\underline{b}=a(s_c)$. Thus, from the monotonicity of $a(\cdot)$ we know that $s_c< b^\star(y_c)$.
We complete the proof by showing that $s_c>\underline{b}+c$. Arguing by contradiction, we suppose that $s_c \leq\underline{b}+c$, which implies that $a(s)$ is well-defined at $s=\underline{b}+c$ and $a(\underline{b}+c)\ge\underline{b}$. However, it follows from the established fact $a(s)<s-c$, that $a(\underline{b}+c)<\underline{b}+c-c=\underline{b}$, which is indeed a contradiction. 
\end{proof}

\begin{proof}[Proof of Theorem \ref{thm3}]
Let $\tilde{c}$ be the level defined by \eqref{tilc}.

\ul{\it Proof of part for $s \geq b^\star(y_c)$.}
By the construction of $m(\cdot)$ and the fact that $\tilde{y}<\hat{y}$, we conclude that the critical value $y_c$ from \eqref{yc2}, satisfies 
\ben
m(y_c)=c<\tilde{c}=m(\tilde{y})
\quad \Leftrightarrow \quad y_c > \tilde{y}
\quad \Leftrightarrow \quad y_c \in\,(\tilde{y},\hat{y}). 
\een 
Using the definition \eqref{eq:m} of $m(\cdot)$, we know that the unique value $y_c$ satisfies \eqref{z_cn}.   
Recall from Theorem \ref{thm0}.(b) that, for $y_c \in(\tilde{y},\hat{y})$, the optimal stopping region $\mc{D}_{y_c}$ for problem \eqref{eq:oldproblem} is given by
\be
\mc{D}_{y_c} = 
[\underline{b}, a^\star(y_c)] \cup [b^\star(y_c),\infty)\,. \label{D^tily}
\ee
By the monotonicity of $\mc{D}_y$ in $y$ (cf. Theorem \ref{thm0}), we have that,
\ben
\text{if} \quad s\ge b^\star(y_c) \,, \quad \text{then} \quad 
[s,\infty)\,\subseteq[b^\star(y_c),\infty)\,\subseteq\mc{D}_{y_c}\subseteq\mc{D}_{s-c}\,.
\een
Hence, by Proposition \ref{prop:compare} we know that $V(x,s;c)\equiv \ol{v}(x;s-c)$ 
for any $(x,s)\in\mc{O}_+$ such that $s\ge b^\star(y_c)$, where $\ol{v}$ admits the expressions in Theorem \ref{thm0} for $y=s-c$. 
The result follows by the straightforward observation that the Neumann condition \eqref{verifeq2} holds for all $s \geq b^\star(y_c)$ (see e.g. Section \ref{sec:51} for similar arguments). 

\vspace{1pt}

\ul{\it Proof of part for $s<b^\star(y_c)$.}
We now focus on the remaining case that $s<b^\star(y_c)$.  
As in Theorems \ref{thm1} and \ref{thm2}, we prove that the take-profit selling strategy $T_{b^\star(y_c)}^+$ still constitutes part of the optimal strategy if $s<b^\star(y_c)$. 
However, we shall prove that the optimal selling strategy is also partially given by a trailing stop type order, which requires selling the asset if its log price drops below some moving threshold $a(s)$ that depends on the running best performance $s$ of the asset log price. 

To be more precise, 
we define the value $\tilde{V}$ of the two-sided exit strategy from an interval $(a, b^\star(y_c))$ by the asset log price process $X$, where $a=a(s)$ is such that $a < s-c$ for a fixed $s\le b^\star(y_c)$. 
In view of the equivalent expression in \eqref{eq:problem2} of our original problem  \eqref{eq:problem1}, the value of this strategy is given by
\begin{equation} \label{tilUdef} 
\tilde{V}(x,s;c,a) := \ex_{x,s}\left[ \exp(-R_{T_a^-\wedge T_{b^\star(y_c)}^+}^c) \, 
\ul{v}(X_{T_a^-\wedge T_{b^\star(y_c)}^+}) \right] .
\end{equation}
We now derive a useful renewal equation satisfied by $\tilde{V}$. 
Since the second component of the process $(X,\overline{X})$ is constant up to time $T_s^+$, we know that $R_{t}^{c} = A_{t}^{s-c}$ for all $t \leq T_s^+$ ($\pr_{x,s}$-a.s.), so we can rewrite \eqref{tilUdef} in the form
\begin{align*} 
\tilde{V}(x,s;c,a) &= 
\ex_{x,s}\left[ \exp(-A_{T_a^-}^{s-c}) \, \ul{v}(X_{T_a^-}) \, \ind_{\{  T_{a}^{-}<T_{s}^{+}\}}\right] + 
\ex_{x,s}\left[ \exp(-A_{T_{s}^+}^{s-c}) \, \tilde{V}(s,s;c,a) \, \ind_{\{  T_{s}^{+}<T_{a}^{-}\}} \right] \,.
\end{align*}
Hence, using  Lemma \ref{lem:I} and \cite[Proposition 4.12]{OmegaRZ}, we conclude that
\begin{align} \label{tilU}
\tilde{V}(x,s;c,a) &= 
\ul{v}(x) + \Delta(x,a;s-c) + \frac{W^{(r,q)}(x,a;s-c)}{W^{(r,q)}(s,a;s-c)} \, \Big( \tilde{V}(s,s;c,a) - \Delta(s,a;s-c) - \ul{v}(s) \Big) 
\,, 
\end{align}
where the non-negative function $\Delta(x,a;y)$ is defined in \eqref{D}. 

\vspace{1pt}

\ul{\it Optimality of the trailing stop $a(s)$.}
We now choose an appropriate threshold $a$ that maximises the value of the two-sided selling strategy $\tilde{V}(x,s;c,a)$ for each fixed $x\le s\le b^\star(y_c)$. 
In order to make it a reasonable candidate for the value function (see also beginning of Section 5), we invoke the principle of continuous (resp., smooth) fit for function $\tilde{V}$ when $X$ has bounded (resp., unbounded) variation. Specifically:

{\it Continuous fit}: 
If the asset log price process $X$ has paths of bounded variation, we start from the point $(x,s)=(a,s) \equiv (a(s),s)$, then by imposing continuous fit at $x=a$, i.e.  $\tilde{V}(a,s;c,a) = \ul{v}(a)$ in \eqref{tilU}, and recalling from \eqref{D} that, $\Delta(x,a;y)=0$ for any $x \leq a\le y$, we obtain 
\begin{align*}
\frac{W^{(r+q)}(0)}{W^{(r,q)}(s,a;s-c)} \, \Big( \tilde{V}(s,s;c,a) - \Delta(s,a;s-c) - \ul{v}(s) \Big) =0
\,. 
\end{align*}
where we used 
\eqref{eq:bigW1} for simplification. 
Because $W^{(r+q)}(0)>0$ (cf. Lemma \ref{lem W}), the above equation is equivalent to
\be
\tilde{V}(s,s;c,a) = \Delta(s,a;s-c) + \ul{v}(s) , \quad \text{for all } s<b^\star(y_c).\label{eq:Ufit}
\ee

{\it Smooth fit}:
If the asset log price process $X$ has paths of unbounded variation, then the fact that $W^{(r+q)}(0)=0$ (cf. Lemma \ref{lem W}) guarantees the continuous fit at $x=a$. However, taking the partial derivative of \eqref{tilU} with respect to $x$, we have
\begin{align*} 
\frac{\p}{\p x} \tilde{V}(a,s;c,a) &= 
\ul{v}'(a) + \frac{W^{(r+q)\prime}(x-a)}{W^{(r,q)}(s,a;s-c)} \, \Big( \tilde{V}(s,s;c,a) - \Delta(s,a;s-c) - \ul{v}(s) \Big) 
\,,\quad\forall  \;a<x<s-c.
\end{align*}
Then, Lemma \ref{lem W} asserts that $W^{(r+q)\prime}(0+)>0$, so by imposing smooth fit at $x=a$, i.e. $\frac{\p}{\p x} \tilde{V}(a,s;c,a) = \ul{v}'(a)$, we again obtain \eqref{eq:Ufit}.

\vspace{0.2cm} 

In both cases, we obtain from \eqref{tilU} and \eqref{eq:Ufit} that the threshold $a=a(s)$ (if it exists) and function $\tilde{V}(\cdot,s;c,a(s))$ satisfy
\begin{multline} \label{tilUfin}
\tilde{V}(x,s;c,a(s)) = \ul{v}(x) + \Delta(x,a(s);s-c) \\ \equiv \ul{v}(x) + \int_{a(s)}^{s-c} W^{(r,q)}(x,w;s-c)\cdot[q\,\underline{v}(w)-\chi(w)]\diff w - \int_{s-c}^{x\vee (s-c)} W^{(r)}(x-w)\cdot\chi(w)\diff w. 
\end{multline}
Then by the above construction, it is easily seen that the mapping $(x,s)\mapsto \tilde{V}(x,s;c,a(s))$ is continuous over $\mc{O}_+$, and $x\mapsto \tilde{V}(x,s;c,a(s))$ is $C^1$ over $(-\infty,s)$ for each fixed $s$ in the unbounded variation case. 
Also, by straightforward calculations, the function $\tilde{V}(\cdot,\cdot;c,a(s))$ from \eqref{tilUfin} (defined in \eqref{tilUdef}) satisfies the Neumann condition \eqref{verifeq2} if and only if $a(\cdot)$ solves the ODE \eqref{eq:ODE}, where the boundary condition follows from the structure of the optimal selling region \eqref{D^tily} when $s=b^\star(y_c)$.
Given that $\{(x,s) \in \mc{O}_+ : x<\ul{b}\}$ is always part of the continuation region $ \mc{O}_+ \setminus \mc{S}_c$ of problem \eqref{eq:problem1} (cf. Proposition  \ref{prop:compare1}), the candidate optimal threshold $a(s)$ must satisfy $a(s) \geq \ul{b}$, which is the final condition imposed in Lemma \ref{lem:as}. 
Notice that, Lemma \ref{lem:as} then implies that the function $a(\cdot)$ is strictly increasing and there exists a unique value $s_c < b^\star(y_c)$, such that $a(s_c) = \ul{b}$.
The above function $\tilde{V}$ is precisely the positive function from the right-hand side of \eqref{eq:thm3} for $s_c \leq s < b^\star(y_c)$.

\ul{\it Proof of sub-part for $s<s_c$.}
In light of the above, the selling strategy in \eqref{tilUdef} (see also  \eqref{tilUfin}) is a candidate only for $s_c \leq s < b^\star(y_c)$, while for all $s<s_c$, it is optimal to simply wait until the asset log price increases to $s_c$ and then follow the optimal strategy ${V}(s_c,s_c;c)$. 
Thus, the expression of the value function in \eqref{eq:thm3} for all $s<s_c$ follows from similar arguments to the ones leading to \eqref{U}--\eqref{Us<zc} and their optimality in the proof of Theorem \ref{thm2}, as soon as we prove the optimality of $\tilde{V}(x,s;c,a(s))$ for all $s_c \leq s < b^\star(y_c)$.

\ul{\it Proof of sub-part for $s_c \leq s < b^\star(y_c)$.}
Therefore, in order to complete the proof, it suffices to show that $\tilde{V}$ satisfies the variational inequality \eqref{verifeq} for all $(x,s)\in\mc{O}_+ \backslash\{(s,s)\}$ such that $s_c \leq s < b^\star(y_c)$.

Firstly, it is seen by construction and the definition \eqref{D} of $\Delta$, that the inequalities involving the infinitesimal generator can be straightforwardly verified, since 
\begin{alignat*}{2}
(\mc{L}-(r+q\ind_{\{x<s-c\}}))\tilde{V}(x,s;c,a(s)) &= 0,  \quad\,&\forall\, a(s)<x<s, \;s_c\le s<b^\star(y_c),\\
(\mc{L}-(r+q\ind_{\{x<s-c\}}))\tilde{V}(x,s;c,a(s)) &= (\mc{L}-(r+q))\underline{v}(x)\le 0,  \quad\,&\forall\, x<a(s), \;s_c\le s<b^\star(y_c).
\end{alignat*}
Hence, it remains to show in the sequel that the inequalities involving the dominance of $\tilde{V}(x,s;c,a(s))$ over the intrinsic value $U(\et^x)$ hold as well, namely 
\ben
\tilde{V}(x,s;c,a(s))> U(\et^x), \quad\forall \;
a(s)<x\le s, \quad  s_c\le s< b^\star(y_c). 
\een
We know from Lemma \ref{lem:as} that 
$$
\tilde{V}(x,s;c,a(s)) = \ul{v}(x) + \Delta(x,a(s);s-c) > U(\et^x), \quad \text{for all } \, x\in(a(s),s-c]\subsetneq[\underline{b},\infty) .
$$ 
We therefore focus on any fixed $x\in(s-c,s]$, for which we have $\frac{\p}{\p s} \tilde{V}(x,s;c,a(s))= \frac{\p}{\p s} \Delta(x,a(s);s-c)$, where
\begin{align*}
\frac{\p}{\p s} \Delta(x,a(s);s-c) = a'(s) \, W^{(r,q)}(x,a(s);s-c) \, (\mc{L}-r-q)\underline{v}(a(s)) + q W^{(r)}(x-s+c) \, \tilde{V}(s-c,s;c,a(s))\,. 
\end{align*}
Using the expression of $\tilde{V}(s-c,s;c,a(s))$ from \eqref{tilUfin} and the ODE \eqref{eq:ODE} solved by $a(\cdot)$, we get 
\ben
\tilde{V}(s-c,s;c,a(s)) = - \frac{a'(s) \, W^{(r,q)}(s,a(s);s-c)}{qW^{(r)}(c)} \, (\mc{L}-r-q)\underline{v}(a(s)) \,. 
\een
Combining all of the above, we obtain that, for all $x\in\,(s-c,s]$ and $s\in(s_c, b^\star(y_c))$,  
\begin{align*}
&\frac{\p}{\p s} \tilde{V}(x,s;c,a(s))\nn\\
&=a'(s) \, W^{(r,q)}(s,a(s);s-c)\bigg(\frac{W^{(r,q)}(x,a(s);s-c)}{W^{(r,q)}(s,a(s);s-c)}-\frac{W^{(r)}(x-s+c)}{W^{(r)}(c)}\bigg) (\mc{L}-r-q)\underline{v}(a(s))\nn\\
&=a'(s) \, W^{(r,q)}(s,a(s);s-c)\bigg(
\ex_x [\exp(-A_{T_{s}^+}^{s-c}) \, \ind_{\{T_{s}^+<T_{a(s)}^-\}}] 
- \ex_x [\exp(-rT_s^+) \, \ind_{\{T_s^+<T_{s-c}^-\}}] \bigg) 
(\mc{L}-r-q)\underline{v}(a(s))\nn\\
&=a'(s) \, W^{(r,q)}(s,a(s);s-c) \, \ex_x[\exp(-A_{T_{s}^+}^{s-c}) \, \ind_{\{T_{s-c}^-<T_{s}^+<T_{a(s)}^-\}}] \, (\mc{L}-r-q)\underline{v}(a(s))<0,
\end{align*}
where we used \eqref{up} and \eqref{uphit} in the second identity, while the last inequality is due to the facts that 
\[(\mc{L}-r-q)\underline{v}(x)=\chi(x)-q\ul{v}(x),\quad\forall x\ge\ul{b}\]
is strictly decreasing and that (see \eqref{eq:fff})
$\chi(\ul{b})-q\ul{v}(\ul{b})=-\half\sigma^2\Phi(r+q)\le 0$.
We therefore see that 
\be \label{tilUmon}
\text{the mapping $s\mapsto \tilde{V}(x,s;c,a(s))$ is strictly decreasing on $[s_c, b^\star(y_c))$ for any fixed $x\in\,(s-c,s]$.} 
\ee
We then argue by contradiction, assuming that there exists a pair $(x_0,s_0)$ for $x_0 \in(s_0-c, s_0]$ and $s_0\in[s_c,b^\star(y_c))$, such that $\tilde{V}(x_0,s_0;c,a(s_0)) \le U(\et^{x_0})$. 
We arrive to a contradiction in both scenarios of $x_0 \in\,(y_c,s_0]$ and $x_0 \in \,(s_0-c, y_c]$. In particular, on one hand,
if $y_c<x_0\le s_0< b^\star(y_c)$, then \eqref{tilUmon} yields that 
\[U(\et^{x_0}) \ge \tilde{V}(x_0,s_0;c,a(s_0))
>\tilde{V}(x_0,b^\star(y_c);c,a(b^\star(y_c)) 
\overset{\eqref{eq:ODE}}{=}\ul{v}(x_0) + \Delta(x_0,a^\star(y_c); y_c).\]
However, combining the above with Theorem \ref{thm0}.(b) for $y=y_c>\tilde{y}$, we get the contradiction  
\[
U(\et^{x_0}) > \ul{v}(x_0) + \Delta(x_0,a^\star(y_c); y_c) \overset{\eqref{vflast}}{=} \ol{v}(x_0;y_c) \ge U(\et^{x_0}).
\]
On the other hand, 
if $s_0-c<x_0 \leq y_c$, then we define $s_1 := x_0+c \in\,(s_0, b^\star(y_c)]$ and use again \eqref{tilUmon} to get the contradiction 
\[
U(\et^{x_0}) \ge \tilde{V}(x_0,s_0;c,a(s_0)) >  \tilde{V}(x_0,s_1;c,a(s_1)) = \ul{v}(x_0) + \Delta(x_0,a(s_1);x_0) \geq U(\et^{x_0}) ,
\]
where the last inequality follows from Lemma \ref{lem:as}. 
In summary we conclude that 
\ben
\tilde{V}(x,s;c,a(s))\equiv \ul{v}(x) + \Delta(x,a(s);s-c)> U(\et^x) , \quad\forall \, (x,s)\in\{(x,s)\in\mc{O}_+\,:\, a(s)< x, \, s_c\le s<b^\star(y_c)\} \,,
\een
which completes the proof.\end{proof}

\appendix
\section{Preliminaries on scale functions}\label{app:pre}

A well-known fluctuation identity of spectrally negative L\'{e}vy
processes (see e.g. \cite[Theorem 8.1]{Kyprianou2006}) is given, 
for $r\geq0$ and $x\in[a, b]$, by 
\begin{align}
\mathbb{E}_{x}[  \et^{-rT_{b}^{+}} \ind_{\{  T_{b}^{+}<T_{a}^{-}\}
}]  =\frac{W^{(r)}(x-a)}{W^{(r)}(b-a)} \,. \label{up} 
\end{align}
The following result from \cite[Theorem 2]{OccupationInterval} provides a generalization of the above case with deterministic discounting $r$ in \eqref{up} to the case with state-dependent discount rate $r+q\ind_{\{X_t<y\}}$, for some fixed $y\in \R$. 
\begin{lem}\label{lem:bigW}
For any $r\ge 0, q>0$, and $x\le b$ with $a\le y\le b$, we have
\begin{align}
\ex_x[\exp(-A_{T_b^+}^y)\ind_{\{T_b^+<T_a^-\}}] 
&=\frac{W^{(r,q)}(x,a;y)}{W^{(r,q)}(b,a;y)}, \label{uphit}
\end{align}
where $A^y$ is given by \eqref{Aold} and we define the non-negative function (see \cite[(6)-(7)]{OccupationInterval})
\begin{align}
W^{(r,q)}(x,a;y)
&:=W^{(r+q)}(x-a)-q\int_{y}^{x\vee y}W^{(r)}(x-z)W^{(r+q)}(z-a)\diff z \label{eq:bigW1} \\
&=W^{(r)}(x-a)+q\int_{a}^{y}W^{(r)}(x-z)W^{(r+q)}(z-a)\,\diff z\,. \label{newW}
\end{align}
\end{lem}

\noindent 
A related result from \cite[Proposition 4.1]{OmegaRZ} (see also \cite[Corollary 2.(ii)]{OccupationInterval} for the proof of a similar result) will also be useful. 
\begin{lem}\label{lem:I}
For any $x<b$, we have 
\be \label{T+b}
\ex_x\big[\exp(-A_{T_b^+}^y)\big]=\frac{\mc{I}^{(r,q)}(x-y)}{\mc{I}^{(r,q)}(b-y)},
\quad \text{where} \;\; \mc{I}^{(r,q)}(\cdot) \;\; \text{is given by \eqref{Irq}.}
\ee
\end{lem}

Finally, the following lemma gives the behaviour of scale functions at $0+$ and $\infty$; see, e.g., \cite[Lemmata 3.1, 3.2, 3.3]{Kuznetsov_2011}, and \cite[(3.13)]{Leung_Yamazaki_2011}.
\begin{lem} \label{lem W}
For any $r>0$,
\begin{align*}
W^{(r)}(0)    =\left\{
\begin{array}
[c]{ll}
0, & \text{unbounded variation},\\
\frac{1}{\gamma}, & \text{bounded variation},
\end{array}
\right. 
W^{(r){\prime}}(0+) =\left\{
\begin{array}
[c]{ll}%
\frac{2}{\sigma^{2}}, & \text{if }\sigma>0,\\
\infty, & \text{if }\sigma=0\text{ and }\Pi(-\infty,0)=\infty,\\
\frac{r+\Pi(-\infty,0)}{\gamma^{2}}, & \text{if }\sigma=0\text{ and }\Pi
(-\infty,0)<\infty,
\end{array}
\right.
\end{align*}
\end{lem}

\section{Technical results}\label{app:tech}

 \begin{lem}\label{lem:comp}
The take-profit selling targets $\ul{b}$, $z_c$ and $z^\star(y)$, defined, respectively, by \eqref{b_}, \eqref{z_c} and \eqref{z*} for all relevant $y$-values in Theorems \ref{thm1} and \ref{thm2}, are strictly decreasing functions of the risk aversion coefficient $\rho \in [0,1)$. 
\end{lem}
\begin{proof}
Suppose that $0\le \rho_1 < \rho_2 < 1$. 
Denote by $z_1^\star(y)$, $\ul{b}_1$ and $z_2^\star(y)$, $\ul{b_2}$ the selling strategies defined by \eqref{z*}, \eqref{b_}, under $\rho_1$ and $\rho_2$, respectively. 

We can firstly observe from the definitions \eqref{b_} and \eqref{z_c} that both $\ul{b} = \ul{b}(\rho)$ and $z_c = z_c(\rho)$ are strictly decreasing in $\rho\in[0,1)$ independently of the case under consideration (cf.\ Section \ref{mainres}).$^{\ref{foot}}$  

Taking into account (from above) that $\ul{b}_1>\ul{b}_2$ holds, 
we can straightforwardly conclude from Theorems \ref{thm1} and \ref{thm2} that $z_1^\star(y) \equiv \ul{b}_1>\ul{b}_2 \equiv z_2^\star(y)$ for all $y \geq \ul{b}_1$.
In view of the inequalities
\[\infty>z^\star_1(-\infty)=\frac{\log\Phi(r)-\log(\Phi(r)-1+\rho_1)}{1-\rho_1}
>\frac{\log\Phi(r)-\log(\Phi(r)-1+\rho_2)}{1-\rho_2}=z^\star_2(-\infty)>0 \]
we claim that $z_1^\star(y) > z_2^\star(y)$ for all $y < \ul{b}_1$.
To prove this by contradiction, we assume that there exists $y<\ul{b}_1$ such that $z_1^\star(y)=z_2^\star(y)=z$ (this suffices as $z_i^\star(y), i=1,2$, is continuous in $y$). Given that $z^\star(y)$ solves \eqref{eq:zstarg}, it follows that 
\[1=e^{(1-\rho_1)z}\left(1-\frac{1-\rho_1}{\Lambda(z-y)}\right)=e^{(1-\rho_2)z}\left(1-\frac{1-\rho_2}{\Lambda(z-y)}\right) .\]
However, firstly note that the function $\rho \mapsto e^{(1-\rho)z}(1-\frac{1-\rho}{\Lambda(z-y)})$ has fixed convexity in $[0,1]$. Secondly, for any fixed $z, y$, observe that $\rho_0=1$ is already a solution. 
Hence, having two more distinct roots $\rho_1, \rho_2 \in [0,1)$ is impossible. This is a contradiction. 
Therefore, we have that $z_1^\star(y)\neq z_2^\star(y)$ for all $y<\ul{b}_1$.
Combining the above results, we can conclude that we eventually have $z_1^\star(y)>z_2^\star(y)$ for all $y$ in their common domain independently of the case under consideration (cf.\ Theorems \ref{thm1} and \ref{thm2} in Section \ref{mainres}).
\end{proof}

\section*{Acknowledgments}
Neofytos Rodosthenous gratefully acknowledges support from EPSRC Grant Number EP/P017193/1.

\bibliographystyle{imsart-nameyear}

\def\cprime{$'$}

\end{document}